\documentclass[twocolumn,aps,preprintnumbers,showpacs,superscriptaddress,groupedaddress]{revtex4-1}

\usepackage{amsfonts}
\usepackage{amsmath}
\usepackage{graphics}
\usepackage{graphicx}
\usepackage{enumerate}
\usepackage{amssymb}
\usepackage{amsthm}
\usepackage{subfigure}

\usepackage{tikz}
\usetikzlibrary{calc,positioning}

\newcommand{\appref}[1]{\hyperref[#1]{{Appendix~\ref*{#1}}}}
\newcommand{\be}{\begin{eqnarray} \begin{aligned}}
\newcommand{\ee}{\end{aligned} \end{eqnarray} }
\newcommand{\benn}{\begin{eqnarray*} \begin{aligned}}
\newcommand{\eenn}{\end{aligned} \end{eqnarray*}}
\newcommand*{\cG}{\mathcal{G}}

\newcommand*{\cP}{\mathcal{P}}
\newcommand*{\cS}{\mathcal{S}}

\newcommand{\bc}{\begin{center}}
\newcommand{\ec}{\end{center}}

\newtheorem{theorem}{Theorem}
\newtheorem{lemma}{Lemma}

\usepackage{amsfonts}

\def\01{\{0,1\}}

\newtheorem{prop}{Proposition}

\newcommand{\calP}{{\mathcal{P}}}
\newcommand{\calT}{{\mathcal{T}}}

\newcommand{\calQ}{{\mathcal{Q}}}

\newcommand{\ZZ}{\mathbb{Z}}
\newcommand{\EE}{\mathbb{E}}
\newcommand{\FF}{\mathbb{F}}
\newtheorem{conj}{Conjecture}

\newcommand{\IBM}{IBM T.J. Watson Research Center, 1101  Kitchawan Road, Yorktown Heights, NY 10598}

\begin{document}

\title{Trading classical and quantum computational resources}

\author{Sergey Bravyi}
\author{Graeme Smith}
\author{John A. Smolin}
\affiliation{\IBM}

\begin{abstract}
We propose examples of a hybrid quantum-classical simulation where
a classical computer assisted by a small  quantum processor 
can efficiently simulate a larger quantum system.
First we consider sparse quantum circuits such that each qubit
participates  in $O(1)$ two-qubit gates. It is shown that 
any sparse circuit on $n+k$ qubits can be simulated
by sparse circuits on $n$ qubits and a classical processing
that takes time $2^{O(k)} poly(n)$.
 Secondly, we study Pauli-based computation (PBC)
where  allowed operations are   non-destructive eigenvalue measurements
of  $n$-qubit Pauli operators. The computation begins by initializing each qubit
in the so-called magic state. This model is known to be equivalent to the universal quantum
computer. We show that 
any PBC on $n+k$ qubits can be simulated by  PBCs on $n$ qubits and a
classical processing that takes time $2^{O(k)} poly(n)$.
 Finally, we propose a purely classical algorithm that can simulate
a PBC on $n$ qubits in a time $2^{\alpha n} poly(n)$
where $\alpha\approx 0.94$.  This improves upon the brute-force simulation method
which takes time $2^n poly(n)$.
Our algorithm exploits the fact that 
 $n$-fold tensor products of magic states admit a low-rank
decomposition into $n$-qubit stabilizer states. 
 \end{abstract}

\maketitle

\section{Introduction}
\label{sec:intro}

Quantum computers promise a substantial speedup
 over classical ones  for certain number-theoretic 
problems and the simulation of quantum systems~\cite{shor1994algorithms,hallgren2007polynomial,lloyd1996universal}.
Experimental efforts to build a quantum computer remain in their
infancy though, limited
to proof-of-principle experiments on a handful of qubits.  In contrast, the design of classical computers is a mature field
offering billions of operations per second in off-the-shelf machines  and petaflops in leading supercomputers.  To prove their worth, quantum computers will have to offer computational solutions that rival the performance of classical supercomputers, a daunting task to be sure. 

Here we study hybrid quantum-classical computation, wherein a small quantum processor is combined with
a large-scale classical computer to jointly solve a computational task. To motivate this problem, imagine 
that a client can access a quantum computer with 100 qubits and essentially perfect quantum gates. 
Such a computer lies in the regime where it is likely to outperform any classical machine (since it would be nearly impossible to emulate classically). Imagine further that  the client wants to implement  a quantum algorithm 
on 101 qubits, but it is  impossible to expand the hardware
to accommodate one extra qubit.
Does the client have any advantage at all from the access to a quantum computer in this scenario? 
Can one divide a quantum algorithm into  subroutines that require less qubits than the entire algorithm?
Can one implement each subroutine separately and combine their outputs on a classical computer? 
These are the main questions addressed in the present paper. 
Put differently, we ask how to  add one {\em virtual qubit} to an existing quantum machine at
the cost of an increased classical and quantum running times, but without modifying
the machine hardware.
More generally, one may ask what is  the cost of adding $k$ virtual qubits to an existing
quantum computer of $n$ qubits and how to
characterize the tradeoff between quantum and classical resources in these settings.

As one may expect, the cost of adding virtual qubits 
varies for different 
computational models. 
Although the circuit-based model of a quantum computer  is the most natural and well-studied,
several alternative models have been proposed, such as 
the measurement-based~\cite{raussendorf2001one} and the adiabatic~\cite{Aharonov2008adiabatic}
quantum computing, as well as the model DQC1 where most of the qubits  are initialized in the maximally
mixed state~\cite{knill1998power}. Our goal is to identify 
quantum computing models which enable efficient addition of  virtual qubits. 
Below we  describe two examples of such models.

We begin with the model based on sparse  quantum circuits. 
Recall that a quantum circuit on $n$ qubits is a collection of gates, 
drawn from some fixed (usually universal) gate set, 
with $n$ input qubits and $n$ output qubits. 
Below we assume that the gate set includes only one-qubit and two-qubit gates.
Let us say that a circuit is  {\em $d$-sparse} if each qubit participates in at most
$d$ two-qubit gates.  We shall be interested in the regime when $d$ is a constant independent of $n$
or when $d$ grows very slowly, say $d\sim \log{(n)}$.
This regime covers  interesting quantum algorithms that can be 
described by low-depth circuits~\cite{cleve2000fast} since any depth-$d$ quantum circuit must be $d$-sparse
(although the converse is generally not true).
It is believed that a constant-depth quantum computation
cannot be efficiently simulated by classical means only~\cite{Terhal2004constant,BremnerIQC2009}.
It is also likely that early applications of quantum 
computers will be based on relatively low-depth circuits
because they impose less stringent requirements on the
qubit coherence times.

Define a $d$-sparse quantum computation, or $d$-SQC, as 
a sequence of the following steps:
(i) initialization of $n$ qubits in the $|0\rangle$ state, (ii) action of  a $d$-sparse quantum circuit,
(iii) measurement of each qubit in the $0,1$ basis, and (iv) classical processing of the measurement
outcomes that returns a single output bit $b_{out}$. We require that the final
classical processing takes time at most $poly(n)$. 
A classical or quantum algorithm is said to simulate a $d$-SQC if it 
computes probability of the output $b_{out}=1$ with a small additive error.
Our first result is the following theorem, which quantifies the cost of adding $k$ virtual qubits
to a $d$-SQC on $n$ qubits. 
\begin{theorem} \label{thm:4}
Suppose $n\ge kd+1$. Then 
any $d$-sparse quantum computation on $n+k$ qubits
can be simulated by  a $(d+3)$-sparse
quantum computation on $n$ qubits 
repeated $2^{O(kd)} $ times and a  classical processing
which takes time $2^{O(kd)} poly(n)$.
\end{theorem} 
The above result is most useful when both $k$ and $d$ are small, for example,
$k=O(1)$ and $d=O(\log{n})$. In this case both quantum and classical running time
of the simulation scale as $poly(n)$. On the other hand, 
we expect that a direct simulation of a $d$-SQC on a classical computer takes a super-polynomial time
(see the discussion above). 
Hence the theorem provides an example when a hybrid quantum-classical simulation is more efficient
than a classical simulation alone. 

The proof of the theorem  exploits the fact that 
any $d$-sparse quantum circuit $U$ acting on a bipartite system $AB$
with $|A|\approx k$ and $|B|\approx n$ 
can be decomposed into a linear combination of $2^{O(kd)}$ 
tensor product terms $V_\alpha\otimes W_\alpha$, where $V_\alpha$
and $W_\alpha$ are $d$-sparse circuits acting on $A$ and $B$ respectively.
We show that the task of simulating $U$ can be reduced to
simulating the smaller circuits $W_\alpha$, as well as computing certain interference 
terms that involve pairs of circuits $W_\alpha,W_\beta$. 
We show that the interference terms can be estimated 
by a simple SWAP test which can be realized by a 
$(d+3)$-sparse computation on $n$ qubits.

Our second model is called Pauli-based computation (PBC).
We begin with a formal definition of the model. 
Let $\calP^n$ be the set of all hermitian Pauli operators on $n$ qubits,
that is, $n$-fold tensor products of single-qubit Pauli operators
$I,X,Y,Z$ with the overall phase factor $\pm 1$.  
A PBC on $n$ qubits is defined as  a sequence of 
elementary steps labeled by integers $t=1,\ldots,n$ where at each step $t$
one performs a non-destructive eigenvalue measurement
of some Pauli operator $P_t\in \calP^n$. 
Let $\sigma_t$ be the measured eigenvalue of $P_t$.
Note that $\sigma_t=\pm 1$ since 
any element of $\calP^n$ squares to one.
We allow the
choice of $P_t$ to be adaptive, that is,   
$P_t$ may depend on all previously measured eigenvalues $\sigma_1,\ldots,\sigma_{t-1}$.
The latter have to be  stored in a classical memory.  The computation begins by initializing 
each qubit  in the so-called magic state 
\[
|H\rangle=\cos{(\pi/8)} |0\rangle + \sin{(\pi/8)} |1\rangle.
\]
Once all Pauli operators $P_1,\ldots,P_n$ have been measured, the final quantum state
is discarded and one is left with a list of measured eigenvalues $\sigma_1,\ldots,\sigma_n$. 
The outcome of a PBC is a single classical bit 
$b_{out}$ obtained by performing a classical
processing of the measured eigenvalues. All classical processing  must take time at most $poly(n)$. 
We shall prove that the computational power of a PBC does not change if one additionally
requires that all Pauli operators $P_1,\ldots,P_n$ pairwise commute (for all measurement outcomes). 
A classical or quantum algorithm is said to simulate a PBC if it 
computes probability of the output $b_{out}=1$ with a small additive error.
An example of a PBC  is shown at Fig.~\ref{fig:pbc}.

The PBC model naturally appears in  fault-tolerant  quantum computing schemes based on
error correcting codes of
stabilizer type~\cite{gottesman1998theory}.
Such codes enable a simple fault-tolerant 
implementation of non-destructive  Pauli measurements on encoded
qubits, for example using the Steane method~\cite{steane1997active}.
Furthermore, topological quantum codes such as the surface code
enable a direct  measurement of certain logical Pauli operators by measuring a properly chosen subset
of physical qubits~\cite{fowler2009high}. 
Several fault-tolerant protocols for preparing encoded magic states
such as $|H\rangle$ have been developed~\cite{BK04,MEK,bravyi2012magic,jones2013multilevel,fowler2013surface}.
PBCs implicitly appeared in the previous work  on quantum fault-tolerance.
Our analysis closely follows the work 
by Campbell and Brown~\cite{campbell2009structure}
who showed that
a certain class of magic state distillation protocols can be implemented by PBCs.

Let us now state our  results. First, we claim that 
a PBC has the same computational power
as the standard circuit-based quantum computing model.
\begin{theorem}
\label{thm:1}
Any  quantum computation in the  circuit-based model with
$n$ qubits and $poly(n)$ gates drawn from the Clifford+T set can be simulated by a PBC
on $m$ qubits, where $m$ is the number of $T$ gates,
and $poly(n)$ classical processing.
\end{theorem}
Recall that the Clifford+T gate set consists of
single-qubit gates 
\[
H=\frac1{\sqrt{2}} \left[ \begin{array}{cc} 1 & 1\\ 1 & -1 \\ \end{array} \right],
\quad S=\left[\begin{array}{cc} 1 & 0\\ 0 & i \\ \end{array} \right],
\quad T=\left[\begin{array}{cc} 1 & 0\\ 0 & e^{i\pi/4} \\ \end{array} \right],
\]
and the two-qubit CNOT gate. This gate set is known to be universal for
quantum computing. 
Secondly, we show that PBCs enable efficient addition of virtual qubits.
\begin{theorem}
\label{thm:2}
A PBC on $n+k$ qubits can be simulated by 
a PBC on $n$ qubits repeated $2^{O(k)}$ times and  a  classical
processing which takes time $2^{O(k)} poly(n)$.
\end{theorem}
Both theorems follow from the fact that 
a generalized PBC that incorporates 
unitary Clifford gates, ancillary stabilizer states (such as $|0\rangle$ or $|+\rangle$),
and has $poly(n)$ measurements 
can be efficiently simulated by the standard PBC defined above. 
To prove Theorem~\ref{thm:1} we convert a given quantum circuit 
on $n$ qubits with $m$ $T$-gates
into a generalized PBC on $n+m$ qubits 
initialized in the $|0^{\otimes n}\rangle \otimes |H^{\otimes m}\rangle$ state.
Each $T$-gate of the circuit is converted into a simple gadget that includes adaptive
Pauli measurements and consumes one copy of the $|H\rangle$ state.
Simulating such generalized PBC by the standard PBC on $m$ qubits proves Theorem~\ref{thm:1}.

To prove Theorem~\ref{thm:2} we represent  $k$ copies of the magic state $|H\rangle$
as a linear combination of $k$-qubit stabilizer states $\phi_\alpha$
such that $|H\rangle\langle H|^{\otimes k}=\sum_\alpha  c_\alpha |\phi_\alpha\rangle\langle \phi_\alpha|$
for some real coefficients $c_\alpha$. The number of terms in this sum is $2^{O(k)}$. 
We  carry out the simulation independently for each  $\alpha$ using a generalized
PBC on $k+n$ qubits initialized in the state $|\phi_\alpha\rangle\otimes |H^{\otimes n}\rangle$
and combine the outcomes on a classical computer.
Finally, we simulate the generalized PBCs by the standard PBCs on $n$ qubits.

Perhaps more surprisingly, we prove that PBCs can be simulated on a classical computer
alone 
more efficiently than one could expect naively. Let us 
first describe a brute-force simulation method based on the matrix-vector multiplication.
Let $\phi_t$ be the $n$-qubit state obtained after measuring the Pauli operators
$P_1,\ldots,P_t$.  
One can store $\phi_t$ in  a
classical memory as a complex vector of size $2^n$. 
Each step of a PBC involves a transformation 
$\phi_t\to \phi_{t+1}$ where
$\phi_{t+1}=(1/2)(I+\sigma_t P_t)\phi_t$.
Since $P_t$ is a Pauli operator, the matrix of $P_t$ in the standard basis
is a permutation matrix modulo phase factors. 
Thus, for a fixed vector $\phi_t$,  one can compute $\phi_{t+1}$ for both choices of $\sigma_t$
 in time $O(2^n)$. Furthermore, one can compute the norm of $\phi_{t+1}$
in time $O(2^n)$ and thus determine the probability of each measurement
outcome $\sigma_t$. By flipping a classical coin 
one can generate a random variable $\sigma_t=\pm 1$ with the desired probability
distribution.   Since any PBC has at most $n$ steps, the overall cost of the classical
simulation is $O(n2^n)$. 
Below we show that this brute force simulation method is not optimal.
\begin{theorem}
\label{thm:3}
Any PBC on $n$ qubits can be simulated classically in time $2^{\alpha n} poly(n)$,
where $\alpha\approx 0.94$. 
\end{theorem}
Our simulation algorithm exploits the fact that tensor products of magic states 
admit a low-rank decomposition into stabilizer states.
Recall that an $n$-qubit state $\phi$ is called a stabilizer state if 
$|\phi\rangle=U|0^{\otimes n}\rangle$ for some $n$-qubit Clifford
operator $U$  --- a product of the elementary gates $H$, $S$, and the CNOT.

Suppose $\psi$ is an arbitrary $n$-qubit state. Define 
a {\em stabilizer rank} of $\psi$  as the smallest integer $\chi$
such that $\psi$ can be written as
$|\psi\rangle=\sum_{\alpha=1}^\chi c_\alpha |\phi_\alpha\rangle$,
where $c_\alpha$ are complex coefficients and $\phi_\alpha$
are $n$-qubit stabilizer states. 
The stabilizer rank of $\psi$ will be denoted $\chi(\psi)$. 
By definition, $1\le \chi(\psi) \le 2^n$ for any $n$-qubit state $\psi$
and $\chi(\psi)=1$ iff $\psi$ is a stabilizer state. 
For example, the magic state $|H\rangle$ has stabilizer rank $\chi(H)=2$, since
$|H\rangle$ is not a stabilizer state itself, but it can be written as a linear combination
of two stabilizer states $|0\rangle$ and $|1\rangle$. 
Furthermore, using the identity 
\[
|H^{\otimes 2}\rangle=\frac12(|00\rangle+|11\rangle)+
\frac1{2\sqrt{2}}(|00\rangle+|01\rangle+|10\rangle-|11\rangle)
\]
one can easily check that $\chi(H^{\otimes 2})=2$.
More generally, let $\chi_n$ be the stabilizer rank of $|H^{\otimes n}\rangle$.
Note that $\chi_{n+m}\le \chi_n \chi_m$ since  a tensor product of two stabilizer
states is a stabilizer state. 
In particular, $\chi_n\le (\chi_2)^{n/2} = 2^{n/2}$. 

The probability to observe measurement outcomes $\sigma_1,\ldots,\sigma_t$
in a PBC implemented up to a step $t$ can be written as
\[
\langle H^{\otimes n}|\Pi|H^{\otimes n}\rangle=\sum_{\alpha,\beta=1}^{\chi_n} \overline{c}_\alpha c_\beta
\langle \phi_\alpha|\Pi|\phi_\beta\rangle
\]
where $\phi_\alpha$ are $n$-qubit stabilizer states, $c_\alpha$ are complex coefficients,
and $\Pi=\prod_{a=1}^t (I + \sigma_aP_a)/2$ is the projector describing
the partially implemented PBC.
We will use a version of the Gottesman-Knill theorem~\cite{aaronson2004improved}
to show that  each term $\langle \phi_\alpha|\Pi|\phi_\beta\rangle$
can be computed  on a classical computer in time $n^3$. 
Since the number of terms is $\chi_n^2$ and the number of steps is at most $n$,
we would be able to simulate a PBC on $n$ qubits   classically
in time $(\chi_n)^2 n^4$.  
Improving upon the brute-force simulation method thus
 requires an upper bound $\chi_n\le 2^{\beta n}$ for some $\beta<1/2$.
We establish such an upper bound with $\beta=\log_2{(7)}/6 \approx 0.468$
by showing that  $\chi_6\le 7$
which implies  $\chi_n\le (\chi_6)^{n/6}\le 7^{n/6}$.
We expect that the scaling in Theorem~\ref{thm:3} can be  improved
by computing  $\chi_n$ for larger values of $n$. 
In Appendix~B we describe a heuristic algorithm for computing 
low-rank decompositions of $|H^{\otimes n}\rangle$ into stabilizer states
which yields the following upper bounds:   
\begin{center}
\begin{tabular}{|c|c|c|c|c|c|}
\hline
$n$ & $\; 2 \; $ & $\; 3\; $ & $\; 4\; $ & $\; 5 \;$ & $\; 6\; $ \\
\hline
$\chi_n\le $ & $2$ & $3$ & $4$ & $6$ & $7$ \\
\hline
\end{tabular}
\end{center}
We believe that these upper bounds are tight. 
A lower bound $\chi_n\ge \Omega(n^{1/2})$ is proved in Appendix~C.

\begin{figure}[htbp]
\includegraphics[width=6cm]{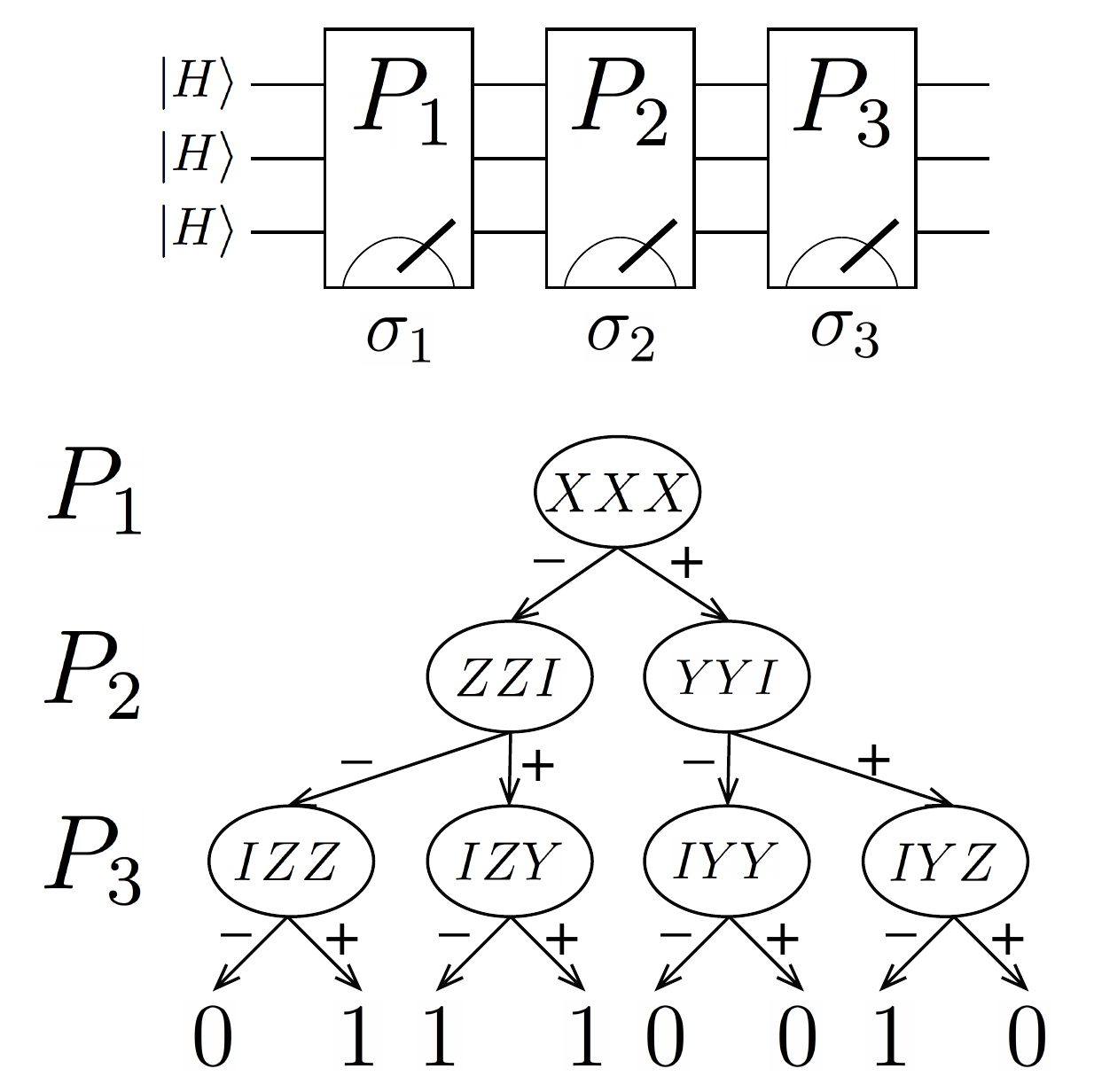}
\caption{Example of a PBC on $n=3$ qubits.
Each step $t$ involves an eigenvalue measurement
of a Pauli operator $P_t$ on $n$ qubits with an outcome $\sigma_t=\pm 1$.
A choice of $P_t$ may depend on the outcomes of all previous measurements.
A PBC on $n$ qubits can be
described by a  binary tree $\calT$ of height $n$ such that internal nodes
of $\calT$ are labeled by $n$-qubit Pauli operators
and leaves of $\calT$ are labeled by $0$ and $1$. 
The latter represent the final output bit $b_{out}$.  We require that label of any node of $\calT$ can be
computed classically in time $poly(n)$. }
\label{fig:pbc}
\end{figure}

\section{Discussion and previous work}
\label{sec:previous}

Classical algorithms for simulation of quantum circuits based on the stabilizer formalism 
have a long history. Notably, 
Aaronson and Gottesman~\cite{aaronson2004improved}
studied  adaptive quantum circuits that contain only a few non-Clifford gates.
Assuming that a circuit contains at most $m$ non-Clifford gates
and that all $n$ qubits are initially  prepared in some stabilizer state,
Ref.~\cite{aaronson2004improved} showed 
how to simulate such a circuit classically  in time $2^{4m}poly(n)$. 
To enable a comparison with our results, assume that all unitary gates belong to the
Clifford+$T$ set. By Theorem~\ref{thm:1}, a  quantum circuit  as above  can be transformed into a PBC
on $m$ qubits, where $m$ is the number of $T$-gates.  
Thus Theorems~\ref{thm:1},\ref{thm:3} provide a classical simulation algorithm with
a running time 
$2^{0.94 m} poly(n)$ which improves upon~\cite{aaronson2004improved}.
In addition, Ref.~\cite{aaronson2004improved} studied 
adaptive quantum circuits composed only of Clifford gates and Pauli measurements
with more general initial states. Assuming that the initial $n$-qubit state
can be written as  a tensor product of some $b$-qubit states,
a quantum circuit as above can be simulated classically in time
$2^{2b+2d}poly(n)$, where $d$ is the total number of measurements~\cite{aaronson2004improved}.

Methods for decomposing  arbitrary states into a linear combination of stabilizer
states aimed at simulation of quantum circuits 
were pioneered by Garcia, Markov, and Cross~\cite{garcia2012efficient,Garcia2014geometry}
who studied  decompositions into pairwise orthogonal stabilizer states (named stabilizer frames).
The latter are more restrictive than the general decompositions analyzed in the present paper.
Furthermore, Refs.~\cite{garcia2012efficient,Garcia2014geometry}
have not studied stabilizer decompositions of magic states.

The simulation algorithm of Theorem~\ref{thm:3}
is conceptually close to  the
matrix multiplication algorithms based on tensor
 decompositions~\cite{strassen1969gaussian,coppersmith1987matrix}.
In this case the analogue of a stabilizer state is 
a product state and the  analogue of a magic state is 
a tripartite entangled state that contains  EPR-type states
shared between each pair of parties, see~\cite{chitambar2008tripartite}
for details.

Efficient classical algorithms for simulation of quantum circuits
in which the initial state can be described by 
a discrete Wigner function taking non-negative values 
were investigated by Veitch et al~\cite{veitch2012negative} and by 
Howard~\cite{howard2014contextuality} et al. 
As was pointed out by Pashayan, Wallman, and Bartlett~\cite{pashayan2015estimating},
such methods can be combined with Monte Carlo sampling techniques
to enable classical simulation of general quantum circuits with the running time
scaling exponentially with the quantity related to the  negativity of the Wigner function.
To enable a comparison between Theorem~\ref{thm:3} and the results of~\cite{pashayan2015estimating}
one can employ a discrete Wigner function representation of 
stabilizer states and  Clifford operations on qubits developed by Delfosse et al~\cite{delfosse2014wigner}.
The latter is applicable only to states with real amplitudes and to Clifford operations
that do not mix $X$-type and $Z$-type Pauli operators (CSS-preserving operations). 
A preliminary analysis shows  that combining the results of Refs.~\cite{pashayan2015estimating,delfosse2014wigner}
yields a classical algorithm for simulating a restricted class of  PBC on $n$ qubits
in time $M^{2n}poly(n) \approx 2^{0.543 n}poly(n)$, where $M=2^{-1}+2^{-1/2}\approx 1.207$
 is the so-called mana of the magic state 
$|H\rangle\langle H|$, see~\cite{pashayan2015estimating,delfosse2014wigner} for details.
The  restriction is that all Pauli operators to be measured are either $X$-type or $Z$-type,
and the measurements cannot be adaptive. Such restricted PBCs are not known to be 
universal for quantum computation. 

Our method of simulating sparse quantum circuits 
has connections to ideas of tensor network representations of quantum circuits 
developed by  Markov and Shi~\cite{markov2008simulating}. 
Indeed, our proof of Theorem~\ref{thm:4} can be interpreted 
as a  particular method of expressing
the acceptance probability of a quantum computation in terms of a contraction of tensors associated with the quantum circuit.  The individual entries of the tensors are then estimated separately  with a smaller quantum computer and then 
added together.  

Let us now discuss some open problems and possible generalizations of our work. 
A natural question is whether the scaling in Theorem~\ref{thm:3} can be improved
if  $|H\rangle$ is replaced  by some other magic state.
By definition,
any  magic state is Clifford-equivalent to one of the states
$|H\rangle$ and $|R\rangle$, where $|R\rangle$ is the $+1$ eigenvector of an operator $(X+Y+Z)/\sqrt{3}$,
see Ref.~\cite{BK04} for details. 
The numerics suggests that 
$|H^{\otimes n}\rangle$ and $|R^{\otimes n}\rangle$ have 
the same stabilizer rank for $n\le 6$. We conjecture that this remains true for all $n$.
Moreover, we  pose the following conjecture which, if true, 
highlights a new optimality property of magic states
in terms of their stabilizer rank.
\begin{conj}
Let $\chi_n$ be the stabilizer rank of $|H^{\otimes n}\rangle$
and $\phi$ be an arbitrary single-qubit state. Then 
\begin{align*}
\chi(\phi^{\otimes n})=1 &&  \mbox{if $\phi$ is a stabilizer state}, \nonumber \\
\chi(\phi^{\otimes n})=\chi_n && \mbox{if $\phi$ is a magic state}, \nonumber \\
\chi(\phi^{\otimes n})>\chi_n  && \mbox{otherwise}.
\end{align*}
\end{conj}
Less formally, the conjecture says that magic states have the smallest possible
stabilizer rank among all non-stabilizer single-qubit states. 

It is also of great interest to understand the asymptotic scaling of the stabilizer
rank $\chi_n$. Assuming that a universal quantum computation cannot be simulated
classically in polynomial time, one infers that $\chi_n$ must grow super-polynomially
in the limit $n\to \infty$. 
However, we were unable to derive such a lower bound directly without using any
assumptions. 
The fact that amplitudes of any stabilizer
state in the standard basis take only $O(1)$ different values
implies a weaker lower   bound $\chi_n\ge \Omega(n^{1/2})$, see Appendix~C.
We conjecture that  in fact $\chi_n\ge 2^{\Omega(n)}$.
Note that if this conjecture is false, that is, $\chi_n\le 2^{o(n)}$,  then constant-depth circuits
in the Clifford+$T$ basis can be simulated classically in a sub-exponential time, which
appears unlikely. Indeed, since such a circuit contains at most $m=O(n)$ $T$-gates,
where $n$ is the number of qubits, 
Theorems~\ref{thm:1},\ref{thm:3} would provide a simulation algorithm with a running time
$\chi_m^2 \cdot poly(n)=2^{o(n)} poly(n)$. (Here we ignore the
complexity of finding the optimal stabilizer decomposition  since it has to be done
only once for each $n$.)

Finally, one may explore generalizations of the stabilizer rank
to approximate decompositions into stabilizer states. 
It should be pointed out that the simulation algorithm
of Theorem~\ref{thm:3} would require approximate stabilizer decompositions 
with a precision at least  $2^{-\Omega(n)}$
since the probability of a particular measurement outcome $\sigma_1,\ldots,\sigma_t$
can be exponentially small in $n$. It is not clear whether such approximate
decompositions would have a rank substantially smaller than
the exact ones.

In the rest of the paper  we prove the theorems stated in the introduction.
From the technical perspective, Theorems~\ref{thm:4},\ref{thm:1},\ref{thm:2}
follow easily  from the definitions and from the previously known results.
On the other hand,
Theorem~\ref{thm:3} and the notion of a stabilizer rank appear to be new.
We analyze sparse quantum circuits in Section~\ref{sec:SQC}.
A classical algorithm for simulation of PBCs and the stabilizer rank of magic states
are discussed in Section~\ref{sec:srank}.
Theorems~\ref{thm:1},\ref{thm:2} are proved in Section~\ref{sec:PBC}.
Appendix~A proves a technical lemma needed to compute  inner products
between stabilizer states. 
Appendix~B describes a numerical method of computing 
low-rank stabilizer decompositions. Appendix~C proves a lower bound
on the stabilizer rank of magic states.


\section{Sparse quantum circuits}
\label{sec:SQC}

In this section we prove Theorem~\ref{thm:4}.
All quantum circuits considered below are defined with respect
to some fixed basis of gates $\cG$. We assume that
any gate in $\cG$ acts on at most two qubits. Furthermore, we assume
that $\cG$ contains all single-qubit Pauli gates $X,Y,Z$, their controlled versions,
the Hadamard gate, and the $\pi/2$ phase shift  $S=|0\rangle\langle 0| + i |1\rangle\langle 1|$.
For example, $\cG$ could be the Clifford+$T$ basis. 
Let $\Sigma^n\equiv \{0,1\}^n$ be the set of $n$-bit binary strings. 
\begin{lemma}
\label{Lemma:expandGates}
Let $U$ be a $d$-sparse quantum circuit on $k+n$ qubits.
Partition the set of qubits as $AB$, where $|A|=k$ and $|B|=n$.
Then 
\begin{align}
\label{U=VW}
U = \sum_{\alpha=1}^\chi c_\alpha V_{\alpha} \otimes W_{\alpha}, \quad \quad \chi\equiv 2^{4kd},
\end{align}
where $V_{\alpha}$ and $W_\alpha$  are $d$-sparse quantum
circuit acting on $A$ and $B$ respectively, 
and $c_\alpha$ are some  complex coefficients such that 
$\sum_{\alpha=1}^\chi |c_\alpha|^2 =1$. 
\end{lemma}

\begin{proof}
Since $U$ is a $d$-sparse circuit, it contains at most $kd$ two-qubit gates
that couple some qubit of $A$ and some qubit of $B$. 
Let $G_1,\ldots,G_m$ be the list of all such gates, where $m\le kd$.
Any  two-qubit gate $G[i,j]$ acting on qubits $i\in A$ and $j\in B$ can be expanded in the Pauli basis as 
$G[i,j]=\sum_{\alpha=1}^{16} c_\alpha P_{\alpha}[i]\otimes P_{\alpha}[j]$,
where $P_\alpha\in \{ I,X,Y,Z\}$ are Pauli operators 
and $c_\alpha$ are some complex coefficients such that  $\sum_\alpha |c_{\alpha}|^2 = 1$.
Applying the above decomposition to each gate $G_1,\ldots,G_m$ 
and, if necessary, appending dummy identity gates to make $m=kd$, 
one arrives at Eq.~(\ref{U=VW}). Note that replacing a two-qubit gate in $U$ by 
a tensor product of two single-qubit Pauli gates cannot increase the sparsity of the circuit.
Thus each term $V_\alpha\otimes W_\alpha$ is a tensor 
product of  two  $d$-sparse circuits. 
\end{proof}
The classical post-processing step can be described by a 
$poly(n)$ classical circuit $f\, : \, \Sigma^{n+k}\to \{0,1\}$.
By definition of the SQC model, the final output
of a computation is a single random bit $b_{out}=f(x)$, where $x\in \Sigma^{n+k}$ is the bit string
obtained by measuring each qubit of a state $U|0^{n+k}\rangle$ in the $0,1$ basis.
Let $\pi(U)$ be the probability of the output $b_{out}=1$, that is, 
\[
\pi(U)=\langle 0^{n+k} |U^\dag \Pi U |0^{n+k}\rangle, 
\]
\begin{equation}
\Pi=\sum_{x\, : \, f(x)=1} |x\rangle\langle x|.
\end{equation}
Let us first show how to estimate the quantity $\pi(U)$ with a small additive error
using  $dk$-sparse circuits on $n+1$ qubits. 
Substituting Eq.~(\ref{U=VW}) into the definition of $\pi(U)$ one gets
\begin{equation}
\label{p2}
\pi(U)=\sum_{y\in \Sigma^k} \sum_{\alpha,\beta=1}^\chi c_\alpha(y) \overline{c_\beta(y)} \langle \phi_\alpha|\Pi(y)|\phi_\beta\rangle,
\end{equation}
where 
\[
c_\alpha(y)=c_\alpha \langle y|V_\alpha|0^k\rangle,
\quad \quad 
|\phi_\alpha\rangle=W_\alpha|0^n\rangle,
\]
and 
\[
\Pi(y)=\sum_{\substack{z\in \Sigma^n\\ f(yz)=1\\ }}\; |z\rangle\langle z|.
\]
We claim that 
each coefficient $c_\alpha(y)$ can be computed exactly in time $O(kd\cdot 2^k)$.
Indeed, we can merge consecutive 
single-qubit gates  of $V_\alpha$ such that
each qubit is acted upon by at most $d$ two-qubit gates and at most $d+1$ single-qubit gates. 
Thus we can assume that the total number of gates in $V_\alpha$ 
is $O(kd$). One can compute the quantity $\langle y|V_\alpha|0^k\rangle$  classically 
in time $O(kd\cdot 2^k)$ by performing matrix-vector multiplication 
for each gate of $V_\alpha$. Furthermore, it is clear from the proof of Lemma~\ref{Lemma:expandGates}
that each coefficient $c_\alpha$ can be computed in time $O(kd)$. 

Consider some fixed triple $(y,\alpha,\beta)$ that appears in the sum Eq.~(\ref{p2}).
Define a controlled-$W$ operator 
\[
\Lambda(W)=|0\rangle\langle 0|\otimes W_\alpha  + |1\rangle\langle 1|\otimes W_\beta.
\]
Define a quantum circuit $R$ acting on $n+1$ qubits that 
consists of the following steps:
(i) initialize $n+1$ qubits in the $|0\rangle$ state, 
(ii) apply $H$ gate to the first qubit, 
(iii) apply $\Lambda(W)$ with the first qubit acting as the control one, (iv) apply $H$ gate to the first qubit, 
(v) measure each qubit in the $0,1$-basis. 
The construction of $R$, illustrated  at Fig.~\ref{fig:swap}, is very similar to the standard
SWAP test, except that we finally measure each qubit. 
Let $b,z$ be the measurement outcomes, where $b=0,1$ and $z\in \Sigma^n$,
see Fig.~\ref{fig:swap}. 
Define a random variable $\sigma'_{y,\alpha,\beta}$  taking values
$\pm 1$ such that $\sigma'_{y,\alpha,\beta}=1$ iff $b=0$ and $f(yz)=1$.
Otherwise $\sigma'_{y,\alpha,\beta}=-1$.
A simple algebra shows that 
\begin{equation}
\label{real_part}
\mathrm{Re}(\langle \phi_\alpha|\Pi(y)|\phi_\beta\rangle )=\EE(\sigma'_{y,\alpha,\beta}),
\end{equation}
that is $\sigma'_{y,\alpha,\beta}$ is an unbiased estimator of the real part
of $\langle \phi_\alpha|\Pi(y)|\phi_\beta\rangle$.
We claim that one can get a sample of $\sigma'_{y,\alpha,\beta}$
by executing a single instance of a
$dk$-sparse quantum computation on $n+1$ qubits (with certain special properties). 
Indeed, by construction, the circuits $W_\alpha$ and $W_\beta$ 
can be obtained from each other by changing some subset of at most $kd$
single-qubit Pauli gates. Thus
the controlled circuit  $\Lambda(W)$ only needs control for at most $kd$
single-qubit Pauli gates. This shows that the control qubit participates in at most
$kd$ two-qubit gates. Furthermore, since all locations where $W_\alpha$ and $W_\beta$
differ from each other originate from two-qubit gates in the initial $d$-sparse circuit
$U$, we conclude that the circuit $R$ has a special property
that  all qubits except for the control one participate in at most $d$
two-qubit gates. 
One can similarly define a random variable $\sigma''_{y,\alpha,\beta}$
such that 
\[
\mathrm{Im}(\langle \phi_\alpha|\Pi(y)|\phi_\beta\rangle )=\EE(\sigma''_{y,\alpha,\beta}).
\]
The only difference is that the $H$ gate in the circuit $R$ must be replaced by $HS$ gate. 
We conclude that 
\[
\langle \phi_\alpha|\Pi(y)|\phi_\beta\rangle = \EE(\sigma'_{y,\alpha,\beta}) + i  \EE(\sigma''_{y,\alpha,\beta}).
\]
Thus the quantity $\pi(U)$ has an unbiased estimator
\[
\xi \equiv \sum_{y\in \Sigma^k} \sum_{\alpha,\beta=1}^\chi c_\alpha(y) \overline{c_\beta(y)}
\left( \sigma'_{y,\alpha,\beta} + i \sigma''_{y,\alpha,\beta}\right),
\]
that is, $\pi(U)=\EE(\xi)$. Using the bounds $|c_\alpha(y)|\le |c_\alpha|$
and $\sum_{\alpha=1}^\chi |c_\alpha|^2=1$ one gets 
\[
 |\xi|\le 2\sum_{y\in \Sigma^k} \sum_{\alpha,\beta=1}^\chi |c_\alpha(y) \overline{c_\beta(y)}|
\le 2^{k+1} \chi
\]
with probability one. 
By Hoeffding's inequality, one can estimate $\EE(\xi)$ with a small additive error
by generating $c2^{2k} \chi^2$ samples of $\xi$ for some constant $c=O(1)$. 
Generating each sample of $\xi$ requires $2^k \chi^2$ samples of the $\sigma$-variables.
Thus one can estimate $\pi(U)$ by repeated applications of $dk$-sparse circuits
on $n+1$ qubits with the number of repetitions scaling as $c2^{3k} \chi^4 = c2^{16kd +3k}=2^{O(kd)}$.

Recall that  the $dk$-sparse circuits $R$ constructed above 
have a very special pattern of sparsity. Namely, all qubits except for one participate in
at most $d$ two-qubit gates, whereas one remaining qubit 
participates in at most $kd$ two-qubit gates. We can distribute the sparsity more evenly among all
$n+1$  qubits by performing a swap gate that changes position of the control qubit after each
application of a control gate (this is possible only if $n$ is sufficiently large, specifically,
if $n\ge kd+1$).  After this modification one obtains
an equivalent circuit which is $(d+3)$-sparse. 

Finally, we can apply exactly the same arguments as above if the 
subsets $A$  and $B$ in Lemma~\ref{Lemma:expandGates}
have size $|A|=k+1$ and $|B|=n-1$. This frees up one extra qubit
that can play the role of the control one in the above construction.
Now we can estimate $\pi(U)$ by repeated applications of $(d+3)$-sparse circuits
on $n$ qubits with the number of repetitions scaling as $c2^{16(k+1)d +3(k+1)}=2^{O(kd)}$.
This completes the proof of Theorem~\ref{thm:4}.

\begin{figure}[htbp]
\includegraphics[width=7cm]{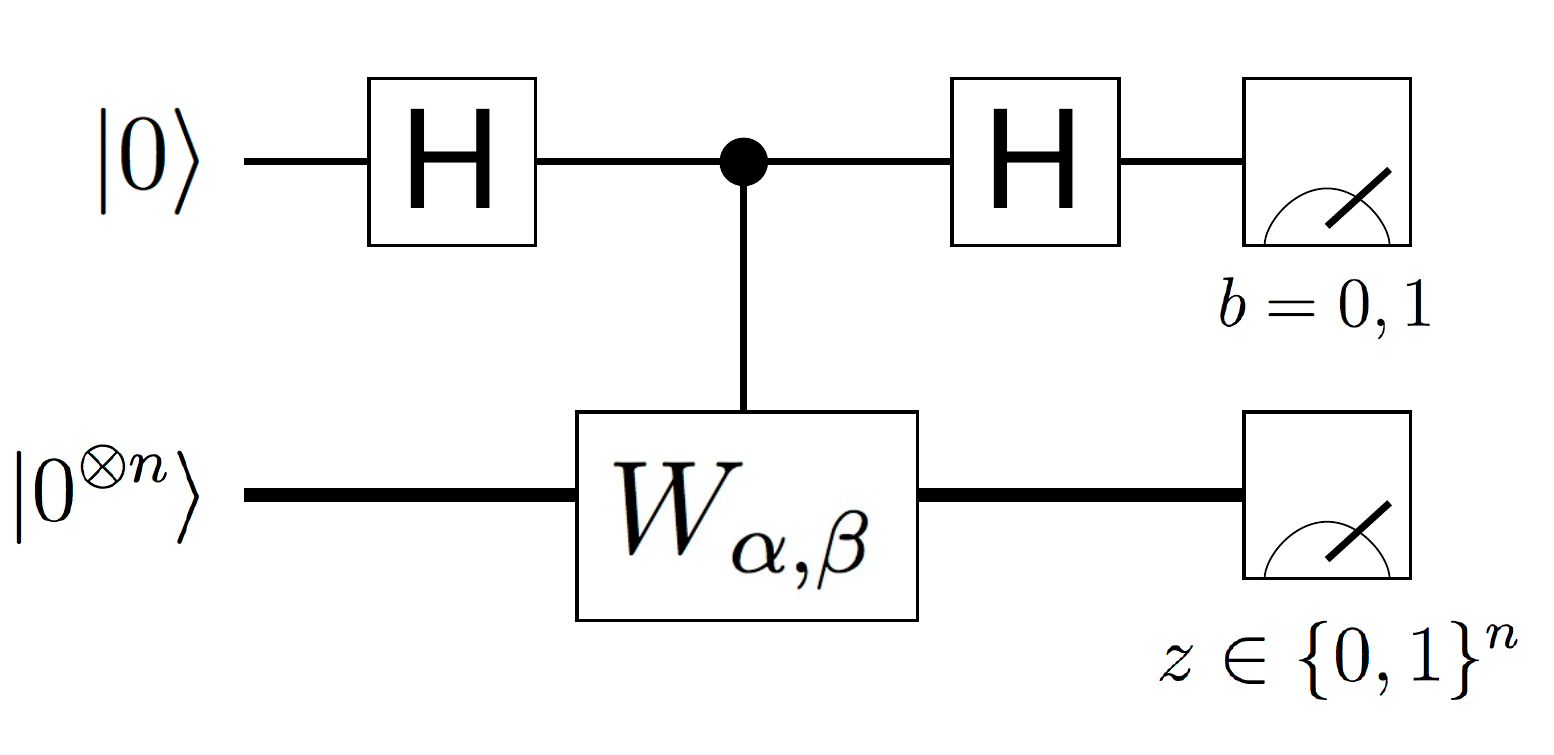}
\caption{Quantum circuit $R$ used to estimate the real part of $\langle \phi_\alpha|\Pi(y)|\phi_\beta\rangle$
in Eq.~(\ref{p2}).
The final output of the circuit is a random variable $\sigma'_{y,\alpha,\beta}=\pm 1$
such that $\sigma'_{y,\alpha,\beta}=1$ iff $b=0$ and $f(yz)=1$, where
$f\, : \, \{0,1\}^{n+k}\to \{0,1\}$ is the Boolean function  describing post-processing step in the original circuit 
$U$ on $n+k$ qubits. We construct a circuit $R$ as above for each triple
$(y,\alpha,\beta)$ with $y\in \{0,1\}^k$ and $\alpha,\beta=1,\ldots,\chi$.
A simple algebra shows that $\sigma'_{y,\alpha,\beta}$
is an unbiased estimator of $\mathrm{Re}(\langle \phi_\alpha|\Pi(y)|\phi_\beta\rangle)$.
}
\label{fig:swap}
\end{figure}


\section{Stabilizer rank and classical simulation of PBC}
\label{sec:srank}

In this section we prove Theorem~\ref{thm:3}.
We begin with an algorithm for computing 
a quantity  $\langle \psi|\Pi|\phi\rangle$,
where $\psi,\phi$ are $n$-qubit stabilizer states  and $\Pi$ is a projector onto the
codespace of some stabilizer code. 
We note that several previous works addressed the problem of computing the
inner product $\langle \psi|\phi\rangle$ between stabilizer states
$\psi,\phi$. 
In particular, Aaronson and  Gottesman~\cite{aaronson2004improved}
showed that  the magnitude $|\langle \psi|\phi\rangle|$can be computed in time $O(n^3)$.
Furthermore,  Garcia, Markov, and Cross~\cite{Garcia2014geometry}
used canonical form of Clifford circuits to compute both the
magnitude and the phase of $\langle \psi|\phi\rangle$  in time $O(n^3)$.
Below we describe a technically different (and somewhat simpler) algorithm 
which is more suited  for computing the quantity $\langle \psi|\Pi|\phi\rangle$ as above.

Let $\ZZ_m\equiv \{0,1,\ldots,m-1\}$
be the cyclic group of order $m$.
A function $f\, : \, \FF_2^n \to \ZZ_8$ is called a {\em degree-two polynomial} if
\[
f(x_1,\ldots,x_n)=f_\emptyset + 2\sum_{a=1}^n f_a x_a + 4 \sum_{1\le a<b\le n} f_{a,b} x_a x_b
\]
where $f_\emptyset\in \ZZ_8$, $f_a\in \ZZ_4$, and $f_{a,b}\in \ZZ_2$
are some constant coefficients. 
Define 
\[
\langle f \rangle = \sum_{x\in \FF_2^n} \omega^{f(x)}, \quad \quad \omega\equiv e^{i\pi/4}.
\]
\begin{lemma}
\label{lemma:sum}
Let $f\, : \, \FF_2^n \to \ZZ_8$ be a degree-two polynomial.
Then either $\langle f\rangle=0$ or $\langle f\rangle=2^{p/2} \omega^m$ for some integer $n\le p\le 2n$
and some $m\in \ZZ_8$. Furthermore, one can compute  $\langle f\rangle$ in time
$O(n^3)$.
\end{lemma}
Since the proof is rather straightforward, we postpone it until  Appendix~A.
It was shown by Dehaene and De Moor~\cite{dehaene2003clifford}
and by Van den Nest~\cite{nest2008classical} that any stabilizer state  $\psi$ of $n$-qubit
can be written (up to a global phase and a normalization) as
\begin{equation}
\label{stab_def}
|\psi\rangle =\sum_{u\in \FF_2^k} \omega^{f(u)} |z+u\Psi\rangle,
\end{equation}
for some degree-two polynomial $f\, : \, \FF_2^k\to \ZZ_8$, some $k\times n$ binary matrix $\Psi$,
and some vector $z\in \FF_2^n$. Here we treat $u$ and $z$ as row vectors.
In the rest of this section we  take Eq.~(\ref{stab_def}) as our definition of a stabilizer state.

Let $\cG\subset \calP^n$ be an abelian group  with $t$ independent generators $P_1,P_2,\ldots,P_t\in \cG$.
Define a projector $\Pi$ onto the $\cG$-invariant subspace,
\[
\Pi=2^{-t} \sum_{P\in \cG} P.
\]
\begin{lemma}
\label{lemma:proj}
The action of $\Pi$ in the computational basis 
can be represented as
\[
\Pi |x\rangle =2^{-t} \sum_{y\in \FF_2^t}  \omega^{g(y)} (-1)^{yBx^T} |x+yA\rangle
\]
for some degree-two polynomial $g\, : \, \FF_2^t \to \ZZ_8$
and some binary matrices $A,B$ of size $t\times n$. 
\end{lemma}
\begin{proof}
Given a binary vector $f\in \FF_2^n$, let $X(f)\in \calP^n$ be 
the  Pauli operator that applies $X$ to each qubit in the support of $f$. 
Define $Z(f)$ in a similar fashion. 
Let $e^k\in \FF_2^t$ be the basis vector which has a single `1' at the
position $k$. The $k$-th generator of $\cG$ can be written as $P_k=i^{c_k} X(e^k A)Z(e^kB)$ for some $c_k\in \ZZ_4$
and some binary matrices $A,B$ of size $t\times n$. In other words,
the $k$-th row of $A$ (of $B$) specifies the $X$-part (the $Z$-part) of $P_k$. 
Choose any vector $y\in \FF_2^t$. 
Then 
\[
P(y)\equiv \prod_{k\, : \, y_k=1} P_k=\omega^{g(y)} X(yA) Z(yB),
\]
where
\[
\omega^{g(y)} = i^{\sum_{k=1}^t c_k y_k} \cdot (-1)^{\sum_{1\le k<l\le t} (BA^T)_{k,l} y_k y_l}.
\]
Clearly, $g\, : \, \FF_2^t\to \ZZ_8$ is a degree-two polynomial. 
Thus
\[
2^t \Pi |x\rangle= \sum_{y\in \FF_2^t} P(y) |x\rangle = \sum_{y\in \FF_2^t} \omega^{g(y)} (-1)^{yBx^T} |x+yA\rangle.
\]
\end{proof}
Consider now a pair of $n$-qubit stabilizer states $\psi,\phi$, where
$\psi$ is defined in Eq.~(\ref{stab_def}) and 
\begin{equation}
\label{paris}
|\phi\rangle =\sum_{v\in \FF_2^m} \omega^{h(v)} |z'+v\Phi \rangle.
\end{equation}
Here $h\, :\, \FF_2^m \to \ZZ_8$ is a degree-two polynomial,
$\Phi$ is a binary matrix of size $m\times n$, and $z'\in \FF_2^n$ is some vector.
 Using Lemma~\ref{lemma:proj} and Eqs.~(\ref{stab_def},\ref{paris}) one gets
\begin{align*}
\langle \psi|\Pi|\phi\rangle
&= 2^{-t}\sum_{u,v,y} \omega^{h(v)-f(u)+g(y)} \\
& \cdot (-1)^{yB(z'+v\Phi)^T} \langle z+u\Psi|z'+v\Phi+yA\rangle.
\end{align*}
Clearly the non-zero terms are those with $z+u\Psi=z'+v\Phi+yA$.
We can enforce this equality by introducing an extra variable 
$x\in \FF_2^n$ such that 
\[
\langle z+u\Psi |z'+v\Phi+yA\rangle
=2^{-n} \sum_{x\in \FF_2^n}  (-1)^{x(z+u\Psi+z'+v\Phi+yA)}.
\]
Then
\begin{equation}
\label{pasadena}
\langle \psi|\Pi|\phi\rangle 
=2^{-n+t}\sum_{u,v,x,y} \omega^{F(u,v,x,y)}=2^{-n+t}\langle F\rangle
\end{equation}
with
\begin{align*}
F(u,v,x,y)&=h(v)-f(u)+g(y)
+4yB(z'+v\Phi)^T \\
&+ 4x(z+u\Psi+z'+v\Phi+yA).
\end{align*}
Note that $F(u,v,x,y)$ is a degree-two polynomial in $k+m+n+t$ variables.
By Lemma~\ref{lemma:sum}, one can compute the sum $\langle F\rangle$
in time $O(k+m+t+n)^3=O(n^3)$. 
Also, Lemma~\ref{lemma:sum} and Eq.~(\ref{pasadena})  implies that $\langle\psi |\Pi|\phi\rangle$ 
takes values $2^{q/2} \omega^j$ for some integer $q$ and $j\in \ZZ_8$. 

Consider now a PBC on $n$ qubits as defined in Section~\ref{sec:PBC}.
Let $t$ be some fixed time step.
Recall that a sequence of measurement outcomes $\sigma_1,\ldots,\sigma_t$
is observed with the probability 
\[
\mathrm{Pr}(\sigma_1,\ldots,\sigma_t)=\langle H^{\otimes n} | \prod_{k=1}^t (1/2) (I+\sigma_k P_k) |H^{\otimes n}\rangle.
\]
Below we shall construct an algorithm that  takes as input 
a step $t$, a sequence of outcomes 
$\sigma_1,\ldots,\sigma_t$
and returns  $\mathrm{Pr}(\sigma_1,\ldots,\sigma_t)$. 
It allows us to compute $\mathrm{Pr}(\sigma_1)$ and 
get a sample of $\sigma_1$ by flipping a coin with a properly chosen bias. 
By calling the algorithm twice  one can also compute conditional  probabilities 
\[
\mathrm{Pr}(\sigma_t| \sigma_1,\ldots,\sigma_{t-1})=\frac{\mathrm{Pr}(\sigma_1,\ldots,\sigma_t)}
{\mathrm{Pr}(\sigma_1,\ldots,\sigma_{t-1})}.
\]
Thus, for fixed variables $\sigma_1,\ldots,\sigma_{t-1}$ 
one can get a sample of $\sigma_t$  by computing the conditional probability 
$\mathrm{Pr}(\sigma_t| \sigma_1,\ldots,\sigma_{t-1})$ and flipping a coin with a 
properly chosen bias. 
The ability to sample the outcomes $\sigma_1,\ldots,\sigma_n$
from the distribution $\mathrm{Pr}(\sigma_1,\ldots,\sigma_n)$
is equivalent to simulating the PBC classically.
Hence it suffices to construct an algorithm
that computes $\mathrm{Pr}(\sigma_1,\ldots,\sigma_t)$.

Suppose we are given some 
integers $k,\chi=O(1)$ and a decomposition
\begin{equation}
\label{london}
|H^{\otimes k}\rangle =\sum_{\alpha=1}^\chi c_\alpha |\phi_\alpha\rangle
\end{equation}
where $\phi_a$ are $k$-qubit stabilizer states
and $c_a$ are complex coefficients.  Suppose also that $n=mk$ for some integer $m$. 
Taking the $m$-fold tensor power of Eq.~(\ref{london}) one gets
\begin{equation}
\label{berlin}
|H^{\otimes n}\rangle=\sum_{\mathbf{a}=1}^{\chi^m} c_{\mathbf{a}} |{\phi}_{\mathbf{a}}\rangle,
\end{equation}
where $\mathbf{a}=(\alpha_1,\ldots,\alpha_m)$, $c_{\mathbf{a}}=c_{\alpha_1} \cdots c_{\alpha_m}$,
and ${\phi}_{\mathbf{a}}= \phi_{\alpha_1}\otimes \cdots \otimes \phi_{\alpha_m}$.
Note that ${\phi}_{\mathbf{a}}$ are stabilizer states and for a given index $\mathbf{a}$ one can compute the standard
form of ${\phi}_{\mathbf{a}}$ as defined in Eq.~(\ref{stab_def}) in time $O(n)$. 
Denoting
\[
\Pi=2^{-t} \prod_{k=1}^t (I+\sigma_k P_k) 
\]
we get 
\begin{equation}
\label{madrid}
\mathrm{Pr}(\sigma_1,\ldots,\sigma_t)=\sum_{\mathbf{a},\mathbf{b}=1}^{\chi^m}
 \overline{c_{\mathbf{a}}} c_{\mathbf{b}} \langle {\phi}_{\mathbf{a}} |\Pi|{\phi}_\mathbf{{b}}\rangle.
\end{equation}
The discussion above implies that each term $\langle {\phi}_{\mathbf{a}} |\Pi|{\phi}_{\mathbf{b}}\rangle$ can be
computed exactly in time $O(n^3)$. Assuming that arithmetic operations with complex
numbers have a unit cost (see Remark~1 below), the probability  $\mathrm{Pr}(\sigma_1,\ldots,\sigma_t)$
can be computed 
in time $O(\chi^{2m} n^3)=O(\chi^{2n/k} n^3)$.

Let us now show an explicit  decomposition Eq.~(\ref{london}) with $k=6$ and $\chi=7$.
This gives an algorithm for computing $\mathrm{Pr}(\sigma_1,\ldots,\sigma_t)$
with a running time $O(7^{n/3} n^3)$ which is enough to prove Theorem~\ref{thm:3}.
It will be more convenient to normalize the magic state
such that 
\[
|H\rangle =|0\rangle + t|1\rangle, \quad \quad t=\tan{(\pi/8)}=\sqrt{2}-1.
\]
Let $B_n=\FF_2^n$ be the set of all $n$-bit strings and
$B_{n,k}\subset B_n$ be the subset of strings with the Hamming weight exactly $k$.
Let $B_n=E_n\cup O_n$,  where $E_n$ and $O_n$ are the subsets of 
even-weight and odd-weight strings respectively.
Given a set of bit strings $S$, we shall write $|S\rangle=\sum_{x\in S}|x\rangle$ for the
uniform superposition of all strings in $S$. 
For example, $|B_{n,0}\rangle=|0^{\otimes n}\rangle$,
$|B_{n,n}\rangle =|1^{\otimes n}\rangle$, and $|H^{\otimes n}\rangle=\sum_{k=0}^n t^k |B_{n,k}\rangle$.
Define also a  state
\[
|K_n\rangle = \sum_{x\in B_n} (-1)^{|x| (|x|-1)/2} \, |x\rangle=\prod_{i<j} \Lambda(Z)_{i,j} |B_n\rangle
\]
Note that $|B_{n,0}\rangle$, $|B_{n,n}\rangle$, $|E_n\rangle$,  $|O_n\rangle$, and $|K_n\rangle$ are stabilizer states
as defined by Eq.~(\ref{stab_def}). 
Define also a pair of graphs $G'=(V',E')$ and $G''=(V'',E'')$ with six vertices shown on Fig.~\ref{fig:graphs}.
The desired stabilizer decomposition of $|H^{\otimes 6}\rangle$ is
\begin{widetext}
 \begin{eqnarray}
\label{H6}
|H^{\otimes 6}\rangle&=&(-16+12\sqrt{2})|B_{6,0}\rangle + (96-68\sqrt{2})|B_{6,6}\rangle
+(10-7\sqrt{2})|E_6\rangle  + (-14+10\sqrt{2})|O_6\rangle   \nonumber \\
&&
+ (7 -5\sqrt{2})Z^{\otimes 6} |K_6\rangle + (10-7\sqrt{2}) |\phi'\rangle + (10-7\sqrt{2})  |\phi''\rangle.
\end{eqnarray}
where
\[
|\phi'\rangle= \prod_{(i,j)\in E'}  \Lambda(Z)_{i,j}  |O_6\rangle 
\quad 
\mbox{and}
\quad
|\phi''\rangle=\prod_{(i,j)\in E''}  \Lambda(Z)_{i,j} |O_6\rangle.
\]
\end{widetext}
\begin{figure}[htbp]
\includegraphics[width=6cm]{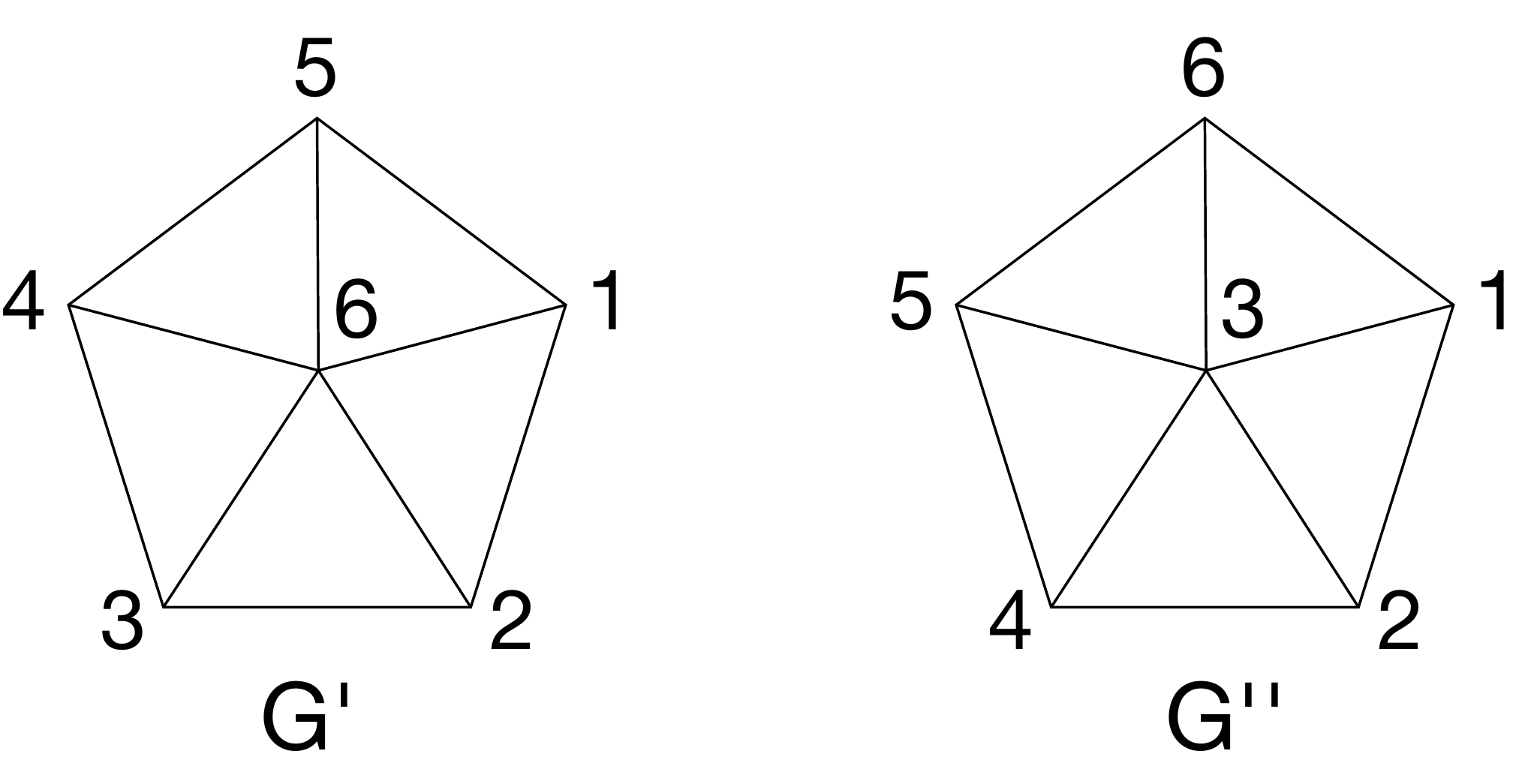}
\caption{Graphs $G'$ and $G''$ used in the definition of stabilizer states $\phi'$ and $\phi''$,
see Eq.~(\ref{H6}).}
\label{fig:graphs}
\end{figure}
This completes the proof of Theorem~\ref{thm:3}.
The numerical method used to find the above decomposition is discussed in Appendix~B.
We conjecture that $k=6$ is the smallest integer such that 
$\chi_k^2<2^k$, see  Section~\ref{sec:intro}.
Accordingly, $k=6$ is likely to be the smallest integer for which the 
 the above  simulation strategy outperforms the brute-force simulation algorithm.

{\em Remark 1:} Let us point out that all coefficients in Eq.~(\ref{H6})
 belong to the 
ring $\ZZ[\sqrt{2}]=\{ p+\sqrt{2}q\, : \, p,q\in \ZZ\}$
known the ring of quadratic integers with  a base two. 
Hence the coefficients $c_{\mathbf{a}}$ in Eq.~(\ref{berlin}) also belong to $\ZZ[\sqrt{2}]$. 
Using Eq.~(\ref{pasadena}) and Lemma~\ref{lemma:sum} we conclude 
that each term 
in Eq.~(\ref{madrid}) has a form $2^{-q} \eta \omega^j$ for some 
integer $0\le q\le n$, some 
$\eta\in \ZZ[\sqrt{2}]$,
and some $j\in \ZZ_8$. Multiplying Eq.~(\ref{madrid})  by 
a suitable power of two we can assume that each term
in Eq.~(\ref{madrid})  has a form $\alpha +i\beta$ where
$\alpha,\beta\in \ZZ[\sqrt{2}]$ (of course we can ignore the 
imaginary part $i\beta$ since $\mathrm{Pr}(\sigma_1,\ldots,\sigma_t)$
is a real number). 
Thus computing the sum in Eq.~(\ref{madrid}) 
only requires arithmetic operations
in the ring $\ZZ[\sqrt{2}]$.

{\em Remark 2:} One can notice that  the first five terms in Eq.~(\ref{H6}) are stabilizer states
symmetric under all permutations of qubits. 
On the other hand, the states $\phi'$ and $\phi''$ break the permutation symmetry. 
Interestingly, we found
that the state $|H^{\otimes 6}\rangle$ does not belong to the subspace
spanned by symmetric stabilizer states of six qubits. Thus any stabilizer decomposition
of $|H^{\otimes 6}\rangle$  must use at least two non-symmetric states.
On the other hand, one can  check that $|H^{\otimes n}\rangle$ belongs
to the subspace spanned by symmetric stabilizer states for $n\le 5$.
The best decompositions that we were able to find for $n\le 5$
are formed by symmetric stabilizer states, see Appendix~B.


\section{Adding virtual qubits to a PBC}
\label{sec:PBC}

In this section we prove Theorems~\ref{thm:1},\ref{thm:2}.
We begin with Theorem~\ref{thm:1}.
Recall that we consider a quantum circuit $U$ on $n$ qubits 
in the Clifford+$T$ basis which contains $m$ $T$-gates.
We assume that all qubits are initialized in the $|0\rangle$ state. 
Each qubit is finally measured in the $0,1$ basis. 
Let us first define a more general version of PBC called
PBC${}^*$ where some subset of qubits 
can be initialized in the $|0\rangle$ state.
Apart from that, definitions of PBC and PBC${}^*$ are the same.
First  we will show that $U$ can be efficiently simulated by PBC${}^*$
on $n+m$ qubits with the initial state $|0^{\otimes n}\rangle \otimes |H^{\otimes m}\rangle$.
Indeed, replace each $T$-gate of $U$ by the gadget shown
on Fig.~\ref{fig:Tgate}.
This gadget uses one ancillary qubit prepared in the magic state
$|T\rangle \sim |0\rangle + e^{i\pi/4}|1\rangle$.
The latter is equivalent to $|H\rangle$ modulo Clifford gates,
$|T\rangle=e^{i\pi/8} H S^{\dag} |H\rangle$.
Let $|\psi\rangle=\alpha |0\rangle + \beta|1\rangle$ be the input
state for the gadget. Let $\sigma_1$ and $\sigma_2$ be the measured
eigenvalues of $ZZ$ and $IX$ operators, see Fig.~\ref{fig:Tgate}.
One can check that 
the gadget outputs a state $\psi_{\sigma_1,\sigma_2}$ where
\begin{eqnarray}
|\psi_{++}\rangle & \sim &  T|\psi\rangle, \nonumber  \\
|\psi_{+-}\rangle &\sim & ZT |\psi\rangle, \nonumber \\
|\psi_{-+}\rangle & \sim & T^{-1} |\psi\rangle, \nonumber \\
|\psi_{--}\rangle &\sim & ZT^{-1} |\psi\rangle. \nonumber
\end{eqnarray}
Furthermore, all  four measurement outcomes are equally likely. 
Applying a correcting Clifford operator
$I,Z,S,ZS$ for the measurement outcomes $++,+-,-+,--$
respectively, one gets the desired $T$ gate. 
Let $U'$ be the circuit obtained from $U$ by replacing each
$T$ gate with the gadget as above. 
\begin{figure}[htbp]
\includegraphics[width=6cm]{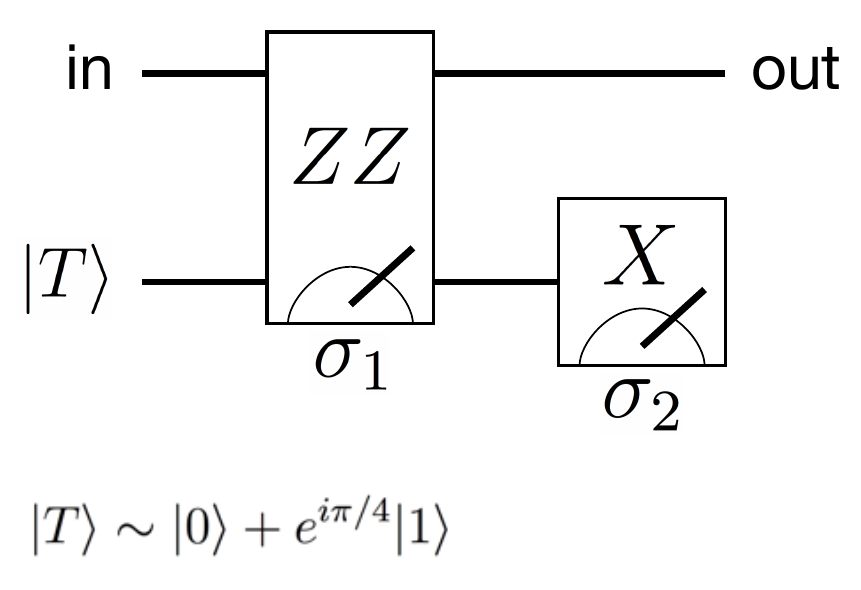}
\caption{Implementation of the 
$T$-gate. }
\label{fig:Tgate}
\end{figure}
The final measurement of $n$ qubits in the $0,1$ basis
is equivalent to a non-destructive eigenvalue measurement
of $Z_1,\ldots,Z_n$ after which  the final state is discarded.
This allows one to commute all Clifford gates of $U'$ towards the end
of the circuit by properly updating which Pauli operator one has to be measured at each step. 
Once a Clifford gate reaches the end of the circuit, it serves no purpose 
and can be discarded.
We conclude that $U$ can be simulated by a PBC${}^*$ on $n+m$ qubits.
Let $P_1,\ldots,P_r\in \calP^{n+m}$ be the Pauli operators that have to be measured.
We can assume that all Pauli operators $P_1,\ldots,P_r$
pairwise commute. Indeed, suppose this is not the case and 
let $t$ be the first time step when $P_t$ anti-commutes
with $P_s$ for some $s<t$. Let $\phi$ be the state reached just before
the measurement of $P_t$. Note that $P_s \phi=\pm \phi$ and 
thus $(I+\sigma_t P_t)\phi=(\sigma_t P_t \pm P_s)\phi$.
One can easily check that an operator $V\equiv (\sigma_t P_t \pm P_s)/\sqrt{2}$ is
a Clifford unitary operator whenever $P_t$ and $P_s$ anticommute. 
This shows that both outcomes $\sigma_t$ have the same probability
and the measurement of $P_t$ has the same effect as drawing
$\sigma_t$  from the uniform distribution and applying 
the Clifford unitary $V$ defined above. Such a unitary $V$ can be commuted towards
the end of the circuit and discarded. 
Hence we can assume that all operators $P_1,\ldots,P_r$ pairwise commute. 
Furthermore, one can append the sequence $P_1,\ldots,P_r$ 
at the beginning with dummy Pauli measurements of $Z_i$
for all qubits $i$ initialized in the $|0\rangle$ state.  Applying the above argument
again one can modify the sequence $P_1,\ldots,P_r$
such that all $P_t$ commute with the dummy measurements, that is, 
any operator $P_t$ acts trivially on the qubits initialized in the $|0\rangle$ state.
Therefore such qubits serve no purpose and can be discarded.
We have shown that the original circuit $U$ can be simulated by a PBC on
$m$ qubits with $r$ steps and pairwise commuting Pauli operators $P_1,\ldots,P_r$.
 Furthermore, since the number of independent
pairwise commuting Pauli operators on $m$ qubits is at most $m$,
we can assume that $r\le m$, that is, the PBC has the standard form. 
This completes the proof of Theorem~\ref{thm:1}.

Let us now prove Theorem~\ref{thm:2}.
Let $\calQ$ be  a fixed PBC  on $n+k$ qubits
and let $p(\calQ)$ be the probability that the final 
outcome of $\calQ$ is $b_{out}=1$. 
Our goal is to approximate $p(\calQ)$ 
on a classical computer assisted by a  PBC on $n$ qubits. 
Suppose one can find a decomposition
\begin{equation}
\label{t2eq1}
|H\rangle\langle H|^{\otimes k}=\sum_{i=1}^\chi \alpha_i |\phi_i\rangle\langle \phi_i|
\end{equation}
for some $k$-qubit stabilizer states $\phi_i$
and some real coefficients $\alpha_i$. 
By linearity, one has 
\begin{equation}
\label{t2e2}
p(\calQ)=\sum_{i=1}^\chi \alpha_i p(\calQ_i),
\end{equation}
where $\calQ_i$ is a PBC-type computation obtained from 
$\calQ$ by initializing the first $k$ qubits in the state $\phi_i$
rather than $|H\rangle^{\otimes k}$.
We note that any stabilizer state $\phi_i$ can be represented
as $|\phi_i\rangle =U_i |0\rangle^{\otimes k}$ for some Clifford
unitary $U_i$. Commuting $U_i$ towards the end of $\calQ_i$
and properly updating which Pauli operator has to be measured at each step
we can assume that $|\phi_i\rangle=|0\rangle^{\otimes k}$ for all $i$.
As we have already showed above,  
such computation $\calQ_i$ is equivalent to a PBC on $n$ qubits.  
Let $b_i$ be the output bit of $\calQ_i$ such that 
$\EE(b_i)=p(\calQ_i)$. 
Define a random variable 
\[
\xi=\sum_{i=1}^\chi \alpha_i b_i.
\]
The above shows that $\xi$ is an unbiased estimator of $p(\calQ)$
and one can generate a sample of $\xi$ by repeating a PBC on $n$ qubits 
$\chi$ times. Since all variables $b_i$ are independent, the 
variance of $\xi$ is bounded as
\begin{equation}
\label{t2eq3}
\sigma^2 \equiv \EE(\xi^2)-\EE(\xi)^2 \le \sum_{i=1}^\chi \alpha_i^2.
\end{equation}
Using the Monte Carlo method one can estimate $p(\calQ)$ with a constant precision 
by generating $M=\min{\{1,O(\sigma^2)\}}$ independent samples of $\xi$.
 Thus the overall cost of adding $k$
virtual qubits is 
\[
C\sim \chi\max{\left(1, \sum_{i=1}^\chi \alpha_i^2\right)}.
\]
It remains to choose a decomposition in Eq.~(\ref{t2eq1}).
One can decompose each copy of
$|H\rangle\langle H|$ as a linear combination of 
stabilizer states using the identity 
\begin{equation}
\label{t2eq4}
|H\rangle\langle H|=\alpha_1 |0\rangle\langle 0| + \alpha_2 |1\rangle\langle 1|
+\alpha_3 |+\rangle\langle +|,
\end{equation} 
where 
\[
\alpha_1=\frac12, \quad \alpha_2=\frac{1-\sqrt{2}}2, \quad \alpha_3=\frac1{\sqrt{2}}
\]
and then take the tensor product decomposition.
Thus  $\chi=3^k$ and $C\sim \chi=2^{O(k)}$.
This completes the proof of Theorem~\ref{thm:2}.


\section*{Appendix~A}

In this section we prove Lemma~\ref{lemma:sum}.
Since the constant term $f_\emptyset$ contributes a multiplicative factor
$\omega^m$ to $\langle f\rangle$, we can assume wlog that $f_\emptyset=0$.
Define coefficients $g_1,\ldots,g_n\in \ZZ_2$ such that 
\[
g_a=\left\{ 
\begin{array}{rcl} 0 &\mbox{if}& f_a=1{\pmod 4},\\
1 &\mbox{if}& f_a=3{\pmod 4},\\
f_a/2 &\mbox{if}& f_a=0,2 {\pmod 4}.\\
\end{array}
\right.
\]
Let $S\subseteq [n]$ be the set of indexes $a$ such that $f_a=1,3 {\pmod 4}$.
A simple algebra shows that
\[
\omega^{f(x)}=i^{\sum_{a\in S} x_a} \cdot (-1)^{g(x)},
\]
where 
\[
g(x)=\sum_{a=1}^n g_a x_a + \sum_{1\le a<b\le n} f_{a,b} x_a x_b.
\]
Let us first assume that $S\neq \emptyset$.
Without loss of generality $S\ni n$ (otherwise permute the variables).
Define a new summation variable $y\in \FF_2^n$ such that 
$y_a=x_a$ for $a=1,\ldots,n-1$ and 
$y_n=\sum_{a\in S} x_a$. 
Note that 
\[
x_a=\left\{ \begin{array}{rcl} y_a &\mbox{if} & a=1,\ldots,n-1,\\
y_n+\sum_{a\in S\backslash n} y_a & \mbox{if} & a=n.\\
\end{array} \right.
\]
Using the identity
\[
i^{\sum_{a\in S} x_a}=i^{\sum_{a\in S} x_a {\pmod 2} }\cdot (-1)^{\sum_{a<b\in S} x_a x_b }
\]
one arrives at
\[
\langle f\rangle = \sum_{y\in \FF_2^n}  i^{y_n} \cdot (-1)^{h(y)}
\]
with 
\begin{align*}
h(y)&=& \sum_{a\notin S} g_a y_a + \sum_{a\in S\backslash n} (g_a+g_n) y_a + g_n y_n \\
&& + \sum_{1\le a<b\le n-1} f_{a,b} y_a y_b  + \sum_{a=1}^{n-1} f_{a,n} y_a y_n \\
&& + \sum_{a=1}^{n-1} \sum_{b\in S\backslash n} f_{a,n} y_a y_b.
\end{align*}
Let us split the sum over $y$ into two terms corresponding to $y_n=0,1$.
We get 
\[
\langle f\rangle = S_0 + i S_1,
\]
 where
\[
S_{\epsilon} = \sum_{z\in \FF_2^{n-1}} (-1)^{h(z,\epsilon)}, \quad \quad \epsilon=0,1.
\]
Using the definition of $h(y)$ one gets 
\[
h(z,\epsilon)=\sum_{1\le a<b\le n-1} H_{a,b}  z_a z_b + L_\epsilon(z)
\]
where $L_\epsilon(z)$  is a linear Boolean function  and
$H$ is a symmetric binary matrix with zero diagonal. 
Importantly, the matrix $H$ does not depend on $\epsilon$. 
It is well-known that any matrix $H$ as above can be transformed into 
a block-diagonal form with all non-zero blocks being
$\left[ \begin{array}{cc} 0 & 1\\ 1 & 0 \\ \end{array}\right]$ by a transformation $H\to V^T H V$, where
$V$ is an invertible binary matrix~\cite{macwilliams1977theory}.
The number of non-zero blocks in $V^T H V$ is $r$, where $2r$ is the rank of $H$
(which is always even). 
Moreover, the matrix $V$ can be computed in time $O(n^3)$
using the standard linear algebra methods~\cite{macwilliams1977theory}.
Performing a change of variable $z\to Vz$ 
and defining new linear functions  $L'_\epsilon(z)=L_\epsilon(Vz)$ one gets
\[
S_\epsilon =\sum_{z\in \FF_2^{n-1}} (-1)^{\sum_{a=1}^r z_{2a-1} z_{2a} + L'_\epsilon(z)}.
\]
Obviously, $S_\epsilon=0$ if $L'_\epsilon(z)$ includes at least one of the variables
$z_a$ with $2r<a\le n-1$. Otherwise one gets
\[
S_\epsilon =2^{n-1-2r} \prod_{a=1}^{r} S_{\epsilon,a},
\]
where 
\[
S_{\epsilon,a}=\sum_{z_{2a-1},z_{2a}=0,1} (-1)^{z_{2a-1} z_{2a} + u(\epsilon,a) z_{2a-1} + v(\epsilon,a) z_{2a}}
\]
for some coefficients $u(\epsilon,a)=0,1$ and $v(\epsilon,a)=0,1$ 
determined by $L'_\epsilon$.
A direct inspection shows that $S_{\epsilon,a}$ takes values $0$ and $\pm 2$.
We conclude that $S_\epsilon$ takes values 
$0$ and $\pm 2^{n-1-r}$. This leaves only nine possible combinations
for  $\langle f\rangle=S_0+iS_1$, Namely,
$\langle f\rangle =0$ (if both $S_0$ and $S_1$ are zero), or
$\langle f\rangle =2^{n-1-r}\omega^{2m}$ for some $m\in \ZZ_4$
(if exactly one of $S_0$ and $S_1$ is non-zero),
or $\langle f\rangle =2^{n-1-r+1/2}\omega^{2m+1}$ for some $m\in \ZZ_4$
(if both $S_0$ and $S_1$ are non-zero).
This is equivalent to the statement of Lemma~\ref{lemma:sum}.
The case when $S=\emptyset$ is completely analogous.


\section*{Appendix~B}

In this section we describe a numerical method for 
computing a low-rank decomposition of a given target state $\phi$
into stabilizer states. 
We shall be mostly interested in the case $|\phi\rangle=|H^{\otimes n}\rangle$.

Let $\cS_n$ be the set of pure $n$-qubit stabilizer states.
Given a target $n$-qubit state $\phi$ and an integer $\chi$ we would like to check 
whether $\phi$ admits a decomposition 
\begin{equation}
\label{A1}
|\phi\rangle=\sum_{a=1}^\chi c_a |\phi_a\rangle
\end{equation}
for some $\phi_1,\ldots,\phi_\chi \in \cS_n$. 
It is known~\cite{aaronson2004improved} that the size of $\cS_n$ grows asymptotically
as  $2^{(1/2+o(1))n^2}$. Thus performing an exhaustive search over all
$\chi$-tuples of $n$-qubit stabilizer states becomes impractical even for
small values of $n$. Instead, we used a Monte Carlo  algorithm 
that performs a  random walk on the set of $\chi$-tuples
$(\phi_1,\ldots,\phi_\chi)\in \cS_n^\chi$ and tries to maximize a suitable objective function
$F(\phi_1,\ldots,\phi_\chi)$. Specifically, we choose
$F(\phi_1,\ldots,\phi_\chi)=\| \Pi \phi\|$, where $\Pi$ is the projector onto the linear
subspace spanned by $\phi_1,\ldots,\phi_\chi$. 
Assuming that $\| \phi\|=1$, the decomposition Eq.~(\ref{A1}) is possible
iff $\max F(\phi_1,\ldots,\phi_\chi)=1$.

We define the random walk on $\cS_n^\chi$ using the Glauber dynamics. 
Let $\beta>0$ be some fixed parameter which has a meaning of the inverse temperature. 
At each step of the walk we randomly choose a state label $a\in \{1,2,\ldots,\chi\}$
and  a Pauli operator $P\in \cP^n$. 
All choices are made with respect to the uniform distribution.
We perform a tentative move $\phi_a \to \phi_a'=c(I+P)\phi_a$, where 
$c$ is a normalizing coefficient.
One can easily check that this move maps stabilizer states to stabilizer states.
If the move increases the value of the 
objective function $F$, we accept the new 
state $\phi_a'$, that is, $\phi_a$ is replaced by $\phi_a'$. Otherwise,
the new state $\phi_a'$ is accepted with a probability
$p_{acc}=\exp{[-\beta (F-F')]}$, where $F$ and $F'$ are the values of the
objective function before and after the move. 
If $(I+P)\phi_a=0$, the move is
rejected right away.
The walk is stopped as long as we observe a tuple of states with $F=1$.
We start with relatively small values $\beta=\beta_{in}$ 
and gradually increase $\beta$ using the geometric sequence 
until it reaches the final value $\beta=\beta_{f}$.
This corresponds to the simulated annealing method. 
For each value of $\beta$ the random walk was repeated for $M\gg 1$ steps.
In practice we used values $\beta_{in}=1$, $\beta_{f}=4000$, and $M=1000$.
The number of annealing steps was chosen as $100$.
Since we worked with relatively small values of $n$, the stabilizer states $\phi_j$
were represented by vectors of size $2^n$. 

Since our target state $|\phi\rangle=|H^{\otimes n}\rangle$
has real amplitudes in the computational basis, one can easily show that the
optimal decomposition Eq.~(\ref{A1}) can be chosen such that 
all stabilizer states $\phi_a$ have real amplitudes as well
(the real part of a stabilizer state is either zero or proportional to a stabilizer state).
Accordingly, we restricted the random walk to the subset of $\cS_n^\chi$
corresponding to real stabilizer states. Clearly, a move $\phi_j \to \phi_j'=c(I+ P)\phi_j$
maps real states to real states if $P$ contains even number of $Y$'s. 
The move was accepted only if this condition is satisfied.

The best decompositions of $|H^{\otimes n}\rangle$ found using this  method
 are shown below. 
Here we use the notations of Section~\ref{sec:srank},
so that $|H\rangle=|0\rangle+(\sqrt{2}-1)|1\rangle$.
\begin{widetext}
\[
|H^{\otimes 2}\rangle =(2-\sqrt{2})|E_2\rangle + (-1+\sqrt{2}) |K_2\rangle.
\]
\[
|H^{\otimes 3}\rangle=(-8+6\sqrt{2})|B_{3,3}\rangle  + (2-\sqrt{2})|E_3\rangle + (-1+\sqrt{2})|K_3\rangle.
\]
\[
|H^{\otimes 4}\rangle=(4-2\sqrt{2})|B_{4,0}\rangle + (20-14\sqrt{2})|B_{4,4}\rangle + (-4+3\sqrt{2})|O_4\rangle
+(-3+2\sqrt{2}) Z^{\otimes 4} |K_4\rangle.
\]
\begin{eqnarray}
|H^{\otimes 5}\rangle&=&(-16+12\sqrt{2})|B_{5,0}\rangle 
+(-40+28\sqrt{2})|B_{5,5}\rangle
+(-4+3\sqrt{2})|O_5\rangle
+(10-7\sqrt{2})|E_5\rangle \nonumber \\
&&+(3-2\sqrt{2}) K |O_5\rangle
+(7-5\sqrt{2}) K|E_5\rangle. \nonumber
\end{eqnarray}
\end{widetext}
Here $K=\prod_{i<j} \Lambda(Z)_{i,j}$ applies cnotrolled-Z to each pair of qubits. 
The stabilizer decomposition of $|H^{\otimes 6}\rangle$ is shown in 
Eq.~(\ref{H6}).  By definition, the number of terms $\chi$ in these
decompositions gives an upper bound on the stabilizer rank $\chi_n$.
We conjecture that all above decompositions and the one in Eq.~(\ref{H6}) are optimal
in the sense that $\chi=\chi_n$. 


\section*{Appendix~C}

Let $\chi_n$ be the stabilizer rank of  $|H^{\otimes n}\rangle$.
Here we  prove a lower bound $\chi_n=\Omega(n^{1/2})$.

Let $\phi$ be a pure $n$-qubit state. Define the $T$-count of $\phi$
denoted $\tau(\phi)$ as the minimum integer $\tau$ such that $\phi$
can be prepared starting from the all-zeros state by a quantum circuit 
composed of Clifford gates,  $T$-gates, and (postselective)
eigenvalue  measurements of Pauli operators,
such that the number of $T$-gates is  at most $\tau$.
We claim that 
\begin{equation}
\label{lb}
 \chi_{\tau(\phi)} \ge \chi(\phi).
\end{equation}
Indeed, as was shown in Section~\ref{sec:PBC}, the $T$-gate can be realized by a gadget that consumes one copy of the
magic state $|H\rangle$ and performs (postselective) Pauli measurements. 
Thus we can prepare $\phi$ starting from $\tau(\phi)$
copies of the magic state $|H\rangle$ by a sequence of (postselective)  Pauli measurement
and Clifford operations.
Since the latter do not increase stabilizer rank, we can write $\phi$ as a linear combination
of $\chi_{\tau(\phi)}$ stabilizer states. This is equivalent to Eq.~(\ref{lb}).

We shall now choose a state $\phi$ will a relatively small $T$-count and a large stabilizer rank. 
Define
\[
|\phi_n\rangle =|\theta_1\otimes \theta_2\otimes \cdots \otimes \theta_n\rangle,
\quad \quad  |\theta_k\rangle=|0\rangle + (2^{k+1}-1)|1\rangle.
\]
\begin{lemma}
\label{lemma:distinct}
The state $\phi_n$ has $2^n$ distinct amplitudes in the computational basis.
\end{lemma}
We postpone the proof of the lemma until the end of the section. 
Let us first show that $\phi_n$ has a large stabilizer rank. 
Indeed, any stabilizer state has $C=O(1)$ distinct amplitudes in the computational basis.
Thus any linear combination of $\chi$ stabilizer states 
has at most $C^\chi$ distinct amplitudes. Applying this to $\phi_n$ one gets
$C^{\chi(\phi_n)} \ge 2^n$, that is, $\chi(\phi_n)=\Omega(n)$. 

Let us now show that $\phi_n$ has a small $T$-count. 
First we claim that the state $\theta_k$ 
has $T$-count $O(k)$. 
 Indeed, we can first prepare
a state $|+\rangle^{\otimes (k+1)} \otimes |0\rangle$ and then apply 
multiple control CNOT gate $\Lambda^{k+1}(X)$ such that the last qubit
is the target one. This creates a state
\[
\sum_{x\in \{0,1\}^{k+1}} |x_1,x_2,\ldots,x_{k+1}\rangle \otimes |x_1 x_2 \cdots x_{k+1}\rangle.
\]
Measuring the first $k+1$ qubits in the $X$-basis and postselecting the outcome $+$
leaves the last qubit in a state 
$(2^{k+1}-1)|0\rangle + |1\rangle$, which 
coincides with $\theta_k$ modulo a bit-flip. One can easily check that 
the multiple control CNOT gate $\Lambda^{k+1}(X)$ can be implemented
using $O(k)$ Toffoli gates. Furthermore, the Toffoli gate
can be implemented using seven $T$-gates~\cite{amy2013meet,gosset2014algorithm}.
Thus $\theta_k$ has $T$-count $O(k)$
and therefore $\phi_n$ has $T$-count $O(\sum_{k=1}^n k)=O(n^2)$.
Substituting this into Eq.~(\ref{lb}) yields $\chi_{n^2}\ge \Omega(n)$, that is,
$\chi_n=\Omega(n^{1/2})$.

\begin{proof}[\bf Proof of Lemma~\ref{lemma:distinct}]
Consider any basis vector $x\in \FF_2^n$. Let $K\subseteq \{1,2,\ldots,n\}$ be the support of $x$. 
Then
\[
\langle x|\phi_n\rangle = \prod_{k\in K} (2^{k+1}-1).
\]
The lemma  follows from the following fact.
\begin{prop}
Suppose $K,M\subseteq [2,\infty)$  are  finite subsets of  integers such that
\begin{equation}
\label{eq1}
\prod_{k\in K} (2^k-1)=\prod_{m\in M} (2^m-1).
\end{equation}
Then $K=M$.
\end{prop}
\begin{proof}
First we claim that 
\begin{equation}
\label{pow2}
\prod_{b\ge a} (1-2^{-b}) > 1- 2^{-a+1} \quad \mbox{for all $a\ge 1$}.
\end{equation}
Indeed, define $x=2^{-a}$. Then 
\begin{align*}
\prod_{b\ge a}(1-2^{-b})^{-1}=\prod_{b\ge 0} (1-x2^{-b})^{-1} \\
=1+\sum_{p=1}^\infty x^p \prod_{q=1}^p (1-2^{-q})^{-1}.
\end{align*}
Define $\xi_p= \prod_{q=1}^p (1-2^{-q})$. One can easily check that 
$\xi_p>\lim_{p\to \infty} \xi_p>1/4$. Since $\xi_1=1/2$, one gets
\[
\prod_{b\ge a}(1-2^{-b})^{-1}<1+2x +4\sum_{p=2}^\infty x^p \le \sum_{p=0}^\infty (2x)^p=(1-2x)^{-1}.
\]
This is equivalent to Eq.~(\ref{pow2}).

Now let $s(K)$ and $s(M)$ be the sum of all elements in $K$ and $M$ respectively. 
Assume wlog that $s(K)\ge s(M)$. Then Eq.~(\ref{eq1}) implies
\begin{align*}
2^{s(K)-s(M)}=\frac{\prod_{m\in M} (1-2^{-m})}{\prod_{k\in K}(1-2^{-k})} \\
\le \frac1{\prod_{k\ge 2} (1-2^{-k})}
<\frac1{1-2^{-1}}=2.
\end{align*}
Here the last inequality follows from Eq.~(\ref{pow2}).
Thus $s(K)=s(M)$ and
\begin{equation}
\label{xi}
\prod_{k\in K} (1-2^{-k})=\prod_{m\in M} (1-2^{-m})\equiv \xi.
\end{equation}
Let $k_1$ and $m_1$ be the smallest elements of $K$ and $M$ respectively. 
Assume wlog that $m_1\ge k_1$. Let us show that in fact $m_1=k_1$.
Indeed, otherwise $m_1\ge k_1+1$. 
Then Eq.~(\ref{xi}) implies $\xi\le 1-2^{-k_1}$ and 
\[
\xi \ge (1-2^{-m_1})\prod_{b\ge m_1+1} (1-2^{-b}) > (1-2^{-m_1})^2.
\]
Here the last inequality follows from Eq.~(\ref{pow2}).
Thus $(1-2^{-m_1})^2 <1-2^{-k_1}$ which implies $m_1<k_1+1$
leading to a contradiction.  We conclude that $k_1=m_1$.
Thus we can cancel the factor $(1-2^{-k_1})$ in both parts of Eq.~(\ref{xi})
and use induction in the number of elements in the largest of the sets $K,M$ 
 to show that $K=M$.
\end{proof}
\end{proof}
Finally, let us sketch an argument that could 
potentially provide a stronger lower bound on $\chi_n$.
 Consider a decomposition
$|H^{\otimes n}\rangle=\sum_{\alpha=1}^{\chi} c_\alpha |\phi_\alpha\rangle$,
where $\phi_\alpha$ are normalized stabilizer states.
We can assume wlog that $\phi_\alpha$ are linearly independent. 
 Define a vector $c=(c_1,\ldots,c_\chi)$
 and a Gram matrix $G_{\alpha,\beta}=\langle \phi_\alpha |\phi_\beta\rangle$.
Then 
$\langle c|G |c\rangle =1$.
Let $g_{min}>0$ be the smallest eigenvalue of $G$. Then $G\ge g_{min} I$ and thus
$\|c\|^2 \le g_{min}^{-1}$.
Let $\delta_n$ be the largest magnitude of the overlap between $|H^{\otimes n}\rangle$ and 
a normalized  $n$-qubit stabilizer state. One can easily check that $\delta_n\le 2^{-\Omega(n)}$.
 The identity 
$1=\sum_{\alpha=1}^\chi c_\alpha^* \langle \phi_\alpha |H^{\otimes n}\rangle$
implies  
$1\le \delta_n \|c\|_1 \le \chi^{1/2} \delta_n \|c\|$.
We conclude that $\chi\ge g_{min} \delta_n^{-2} \ge g_{min} 2^{\Omega(n)}$.
This proves that $\chi\ge 2^{\Omega(n)}$ in the special case when 
all states $\phi_\alpha$ are pairwise orthogonal, that is, $g_{min}=1$.

\section*{Acknowledgments}

SB thanks Martin Roetteler and Jon Yard
for helpful discussions on stabilizer rank of magic states. 
The authors acknowledge  NSF Grant CCF-1110941.


\begin{thebibliography}{34}
\expandafter\ifx\csname natexlab\endcsname\relax\def\natexlab#1{#1}\fi
\expandafter\ifx\csname bibnamefont\endcsname\relax
  \def\bibnamefont#1{#1}\fi
\expandafter\ifx\csname bibfnamefont\endcsname\relax
  \def\bibfnamefont#1{#1}\fi
\expandafter\ifx\csname citenamefont\endcsname\relax
  \def\citenamefont#1{#1}\fi
\expandafter\ifx\csname url\endcsname\relax
  \def\url#1{\texttt{#1}}\fi
\expandafter\ifx\csname urlprefix\endcsname\relax\def\urlprefix{URL }\fi
\providecommand{\bibinfo}[2]{#2}
\providecommand{\eprint}[2][]{\url{#2}}

\bibitem[{\citenamefont{Shor}(1994)}]{shor1994algorithms}
\bibinfo{author}{\bibfnamefont{P.~W.} \bibnamefont{Shor}},
  \bibinfo{journal}{Proceedings of the 35th Annual Symposium on Foundations of
  Computer Science} pp. \bibinfo{pages}{124--134} (\bibinfo{year}{1994}).

\bibitem[{\citenamefont{Hallgren}(2007)}]{hallgren2007polynomial}
\bibinfo{author}{\bibfnamefont{S.}~\bibnamefont{Hallgren}},
  \bibinfo{journal}{Journal of the ACM (JACM)} \textbf{\bibinfo{volume}{54}},
  \bibinfo{pages}{4} (\bibinfo{year}{2007}).

\bibitem[{\citenamefont{Lloyd}(1996)}]{lloyd1996universal}
\bibinfo{author}{\bibfnamefont{S.}~\bibnamefont{Lloyd}},
  \bibinfo{journal}{Science} \textbf{\bibinfo{volume}{273}},
  \bibinfo{pages}{1073} (\bibinfo{year}{1996}).

\bibitem[{\citenamefont{Raussendorf and Briegel}(2001)}]{raussendorf2001one}
\bibinfo{author}{\bibfnamefont{R.}~\bibnamefont{Raussendorf}} \bibnamefont{and}
  \bibinfo{author}{\bibfnamefont{H.~J.} \bibnamefont{Briegel}},
  \bibinfo{journal}{Phys. Rev. Lett.} \textbf{\bibinfo{volume}{86}},
  \bibinfo{pages}{5188} (\bibinfo{year}{2001}).

\bibitem[{\citenamefont{Aharonov et~al.}(2007)\citenamefont{Aharonov, Van~Dam,
  Kempe, Landau, Lloyd, and Regev}}]{Aharonov2008adiabatic}
\bibinfo{author}{\bibfnamefont{D.}~\bibnamefont{Aharonov}},
  \bibinfo{author}{\bibfnamefont{W.}~\bibnamefont{Van~Dam}},
  \bibinfo{author}{\bibfnamefont{J.}~\bibnamefont{Kempe}},
  \bibinfo{author}{\bibfnamefont{Z.}~\bibnamefont{Landau}},
  \bibinfo{author}{\bibfnamefont{S.}~\bibnamefont{Lloyd}}, \bibnamefont{and}
  \bibinfo{author}{\bibfnamefont{O.}~\bibnamefont{Regev}},
  \bibinfo{journal}{SIAM J. of Computing} \textbf{\bibinfo{volume}{37}},
  \bibinfo{pages}{166} (\bibinfo{year}{2007}).

\bibitem[{\citenamefont{Knill and Laflamme}(1998)}]{knill1998power}
\bibinfo{author}{\bibfnamefont{E.}~\bibnamefont{Knill}} \bibnamefont{and}
  \bibinfo{author}{\bibfnamefont{R.}~\bibnamefont{Laflamme}},
  \bibinfo{journal}{Phys. Rev. Lett.} \textbf{\bibinfo{volume}{81}},
  \bibinfo{pages}{5672} (\bibinfo{year}{1998}).

\bibitem[{\citenamefont{Cleve and Watrous}(2000)}]{cleve2000fast}
\bibinfo{author}{\bibfnamefont{R.}~\bibnamefont{Cleve}} \bibnamefont{and}
  \bibinfo{author}{\bibfnamefont{J.}~\bibnamefont{Watrous}},
  \bibinfo{journal}{Proceedings of the 41st Annual Symposium on Foundations of
  Computer Science} pp. \bibinfo{pages}{526--536} (\bibinfo{year}{2000}).

\bibitem[{\citenamefont{Terhal and DiVincenzo}(2004)}]{Terhal2004constant}
\bibinfo{author}{\bibfnamefont{B.}~\bibnamefont{Terhal}} \bibnamefont{and}
  \bibinfo{author}{\bibfnamefont{D.}~\bibnamefont{DiVincenzo}},
  \bibinfo{journal}{Quant. Inf. Comp.} \textbf{\bibinfo{volume}{4}},
  \bibinfo{pages}{134} (\bibinfo{year}{2004}).

\bibitem[{\citenamefont{Shepherd and Bremner}(2009)}]{BremnerIQC2009}
\bibinfo{author}{\bibfnamefont{D.}~\bibnamefont{Shepherd}} \bibnamefont{and}
  \bibinfo{author}{\bibfnamefont{M.}~\bibnamefont{Bremner}},
  \bibinfo{journal}{Proc. R. Soc. A} \textbf{\bibinfo{volume}{465}},
  \bibinfo{pages}{1413} (\bibinfo{year}{2009}).

\bibitem[{\citenamefont{Gottesman}(1998)}]{gottesman1998theory}
\bibinfo{author}{\bibfnamefont{D.}~\bibnamefont{Gottesman}},
  \bibinfo{journal}{Phys. Rev. A} \textbf{\bibinfo{volume}{57}},
  \bibinfo{pages}{127} (\bibinfo{year}{1998}).

\bibitem[{\citenamefont{Steane}(1997)}]{steane1997active}
\bibinfo{author}{\bibfnamefont{A.~M.} \bibnamefont{Steane}},
  \bibinfo{journal}{Phys. Rev. Lett.} \textbf{\bibinfo{volume}{78}},
  \bibinfo{pages}{2252} (\bibinfo{year}{1997}).

\bibitem[{\citenamefont{Fowler et~al.}(2009)\citenamefont{Fowler, Stephens, and
  Groszkowski}}]{fowler2009high}
\bibinfo{author}{\bibfnamefont{A.}~\bibnamefont{Fowler}},
  \bibinfo{author}{\bibfnamefont{A.}~\bibnamefont{Stephens}}, \bibnamefont{and}
  \bibinfo{author}{\bibfnamefont{P.}~\bibnamefont{Groszkowski}},
  \bibinfo{journal}{Phys. Rev. A} \textbf{\bibinfo{volume}{80}},
  \bibinfo{pages}{052312} (\bibinfo{year}{2009}).

\bibitem[{\citenamefont{Bravyi and Kitaev}(2005)}]{BK04}
\bibinfo{author}{\bibfnamefont{S.}~\bibnamefont{Bravyi}} \bibnamefont{and}
  \bibinfo{author}{\bibfnamefont{A.}~\bibnamefont{Kitaev}},
  \bibinfo{journal}{Phys. Rev. A} \textbf{\bibinfo{volume}{71}},
  \bibinfo{pages}{022316} (\bibinfo{year}{2005}).

\bibitem[{\citenamefont{Meier et~al.}(2012)\citenamefont{Meier, Eastin, and
  Knill}}]{MEK}
\bibinfo{author}{\bibfnamefont{A.}~\bibnamefont{Meier}},
  \bibinfo{author}{\bibfnamefont{B.}~\bibnamefont{Eastin}}, \bibnamefont{and}
  \bibinfo{author}{\bibfnamefont{E.}~\bibnamefont{Knill}},
  \bibinfo{journal}{arXiv:1204.4221}  (\bibinfo{year}{2012}).

\bibitem[{\citenamefont{Bravyi and Haah}(2012)}]{bravyi2012magic}
\bibinfo{author}{\bibfnamefont{S.}~\bibnamefont{Bravyi}} \bibnamefont{and}
  \bibinfo{author}{\bibfnamefont{J.}~\bibnamefont{Haah}},
  \bibinfo{journal}{Phys. Rev. A} \textbf{\bibinfo{volume}{86}},
  \bibinfo{pages}{052329} (\bibinfo{year}{2012}).

\bibitem[{\citenamefont{Jones}(2013)}]{jones2013multilevel}
\bibinfo{author}{\bibfnamefont{C.}~\bibnamefont{Jones}},
  \bibinfo{journal}{Phys. Rev. A} \textbf{\bibinfo{volume}{87}},
  \bibinfo{pages}{042305} (\bibinfo{year}{2013}).

\bibitem[{\citenamefont{Fowler et~al.}(2013)\citenamefont{Fowler, Devitt, and
  Jones}}]{fowler2013surface}
\bibinfo{author}{\bibfnamefont{A.}~\bibnamefont{Fowler}},
  \bibinfo{author}{\bibfnamefont{S.}~\bibnamefont{Devitt}}, \bibnamefont{and}
  \bibinfo{author}{\bibfnamefont{C.}~\bibnamefont{Jones}},
  \bibinfo{journal}{Scientific reports} \textbf{\bibinfo{volume}{3}},
  \bibinfo{pages}{1939} (\bibinfo{year}{2013}).

\bibitem[{\citenamefont{Campbell and Browne}(2009)}]{campbell2009structure}
\bibinfo{author}{\bibfnamefont{E.}~\bibnamefont{Campbell}} \bibnamefont{and}
  \bibinfo{author}{\bibfnamefont{D.}~\bibnamefont{Browne}},
  \bibinfo{journal}{arXiv:0908.0838}  (\bibinfo{year}{2009}).

\bibitem[{\citenamefont{Aaronson and Gottesman}(2004)}]{aaronson2004improved}
\bibinfo{author}{\bibfnamefont{S.}~\bibnamefont{Aaronson}} \bibnamefont{and}
  \bibinfo{author}{\bibfnamefont{D.}~\bibnamefont{Gottesman}},
  \bibinfo{journal}{Phys. Rev. A} \textbf{\bibinfo{volume}{70}},
  \bibinfo{pages}{052328} (\bibinfo{year}{2004}).

\bibitem[{\citenamefont{Garcia et~al.}(2012)\citenamefont{Garcia, Markov, and
  Cross}}]{garcia2012efficient}
\bibinfo{author}{\bibfnamefont{H.}~\bibnamefont{Garcia}},
  \bibinfo{author}{\bibfnamefont{I.}~\bibnamefont{Markov}}, \bibnamefont{and}
  \bibinfo{author}{\bibfnamefont{A.}~\bibnamefont{Cross}},
  \bibinfo{journal}{arXiv preprint arXiv:1210.6646}  (\bibinfo{year}{2012}).

\bibitem[{\citenamefont{Garc\'{\i}a et~al.}(2014)\citenamefont{Garc\'{\i}a,
  Markov, and Cross}}]{Garcia2014geometry}
\bibinfo{author}{\bibfnamefont{H.~J.} \bibnamefont{Garc\'{\i}a}},
  \bibinfo{author}{\bibfnamefont{I.}~\bibnamefont{Markov}}, \bibnamefont{and}
  \bibinfo{author}{\bibfnamefont{A.}~\bibnamefont{Cross}},
  \bibinfo{journal}{Quant. Inf. and Comp.} \textbf{\bibinfo{volume}{14}},
  \bibinfo{pages}{683} (\bibinfo{year}{2014}).

\bibitem[{\citenamefont{Strassen}(1969)}]{strassen1969gaussian}
\bibinfo{author}{\bibfnamefont{V.}~\bibnamefont{Strassen}},
  \bibinfo{journal}{Numerische Mathematik} \textbf{\bibinfo{volume}{13}},
  \bibinfo{pages}{354} (\bibinfo{year}{1969}).

\bibitem[{\citenamefont{Coppersmith and
  Winograd}(1987)}]{coppersmith1987matrix}
\bibinfo{author}{\bibfnamefont{D.}~\bibnamefont{Coppersmith}} \bibnamefont{and}
  \bibinfo{author}{\bibfnamefont{S.}~\bibnamefont{Winograd}},
  \bibinfo{journal}{Proceedings of the 19th Annual Symposium on Foundations of
  Computer Science} pp. \bibinfo{pages}{1--6} (\bibinfo{year}{1987}).

\bibitem[{\citenamefont{Chitambar et~al.}(2008)\citenamefont{Chitambar, Duan,
  and Shi}}]{chitambar2008tripartite}
\bibinfo{author}{\bibfnamefont{E.}~\bibnamefont{Chitambar}},
  \bibinfo{author}{\bibfnamefont{R.}~\bibnamefont{Duan}}, \bibnamefont{and}
  \bibinfo{author}{\bibfnamefont{Y.}~\bibnamefont{Shi}},
  \bibinfo{journal}{Phys. Rev. Lett.} \textbf{\bibinfo{volume}{101}},
  \bibinfo{pages}{140502} (\bibinfo{year}{2008}).

\bibitem[{\citenamefont{Veitch et~al.}(2012)\citenamefont{Veitch, Ferrie,
  Gross, and Emerson}}]{veitch2012negative}
\bibinfo{author}{\bibfnamefont{V.}~\bibnamefont{Veitch}},
  \bibinfo{author}{\bibfnamefont{C.}~\bibnamefont{Ferrie}},
  \bibinfo{author}{\bibfnamefont{D.}~\bibnamefont{Gross}}, \bibnamefont{and}
  \bibinfo{author}{\bibfnamefont{J.}~\bibnamefont{Emerson}},
  \bibinfo{journal}{New J. Phys.} \textbf{\bibinfo{volume}{14}},
  \bibinfo{pages}{113011} (\bibinfo{year}{2012}).

\bibitem[{\citenamefont{Howard et~al.}(2014)\citenamefont{Howard, Wallman,
  Veitch, and Emerson}}]{howard2014contextuality}
\bibinfo{author}{\bibfnamefont{M.}~\bibnamefont{Howard}},
  \bibinfo{author}{\bibfnamefont{J.}~\bibnamefont{Wallman}},
  \bibinfo{author}{\bibfnamefont{V.}~\bibnamefont{Veitch}}, \bibnamefont{and}
  \bibinfo{author}{\bibfnamefont{J.}~\bibnamefont{Emerson}},
  \bibinfo{journal}{Nature} \textbf{\bibinfo{volume}{510}},
  \bibinfo{pages}{351} (\bibinfo{year}{2014}).

\bibitem[{\citenamefont{Pashayan et~al.}(2015)\citenamefont{Pashayan, Wallman,
  and Bartlett}}]{pashayan2015estimating}
\bibinfo{author}{\bibfnamefont{H.}~\bibnamefont{Pashayan}},
  \bibinfo{author}{\bibfnamefont{J.}~\bibnamefont{Wallman}}, \bibnamefont{and}
  \bibinfo{author}{\bibfnamefont{S.}~\bibnamefont{Bartlett}},
  \bibinfo{journal}{arXiv preprint arXiv:1503.07525}  (\bibinfo{year}{2015}).

\bibitem[{\citenamefont{Delfosse et~al.}(2014)\citenamefont{Delfosse, Guerin,
  Bian, and Raussendorf}}]{delfosse2014wigner}
\bibinfo{author}{\bibfnamefont{N.}~\bibnamefont{Delfosse}},
  \bibinfo{author}{\bibfnamefont{P.}~\bibnamefont{Guerin}},
  \bibinfo{author}{\bibfnamefont{J.}~\bibnamefont{Bian}}, \bibnamefont{and}
  \bibinfo{author}{\bibfnamefont{R.}~\bibnamefont{Raussendorf}},
  \bibinfo{journal}{arXiv preprint arXiv:1409.5170}  (\bibinfo{year}{2014}).

\bibitem[{\citenamefont{Markov and Shi}(2008)}]{markov2008simulating}
\bibinfo{author}{\bibfnamefont{I.}~\bibnamefont{Markov}} \bibnamefont{and}
  \bibinfo{author}{\bibfnamefont{Y.}~\bibnamefont{Shi}}, \bibinfo{journal}{SIAM
  J. on Comp.} \textbf{\bibinfo{volume}{38}}, \bibinfo{pages}{963}
  (\bibinfo{year}{2008}).

\bibitem[{\citenamefont{Dehaene and De~Moor}(2003)}]{dehaene2003clifford}
\bibinfo{author}{\bibfnamefont{J.}~\bibnamefont{Dehaene}} \bibnamefont{and}
  \bibinfo{author}{\bibfnamefont{B.}~\bibnamefont{De~Moor}},
  \bibinfo{journal}{Phys. Rev. A} \textbf{\bibinfo{volume}{68}},
  \bibinfo{pages}{042318} (\bibinfo{year}{2003}).

\bibitem[{\citenamefont{Van~den Nest}(2010)}]{nest2008classical}
\bibinfo{author}{\bibfnamefont{M.}~\bibnamefont{Van~den Nest}},
  \bibinfo{journal}{Quant. Inf. Comp.} \textbf{\bibinfo{volume}{10}},
  \bibinfo{pages}{0258} (\bibinfo{year}{2010}).

\bibitem[{\citenamefont{MacWilliams and Sloane}(1977)}]{macwilliams1977theory}
\bibinfo{author}{\bibfnamefont{F.}~\bibnamefont{MacWilliams}} \bibnamefont{and}
  \bibinfo{author}{\bibfnamefont{N.}~\bibnamefont{Sloane}},
  \emph{\bibinfo{title}{The theory of error correcting codes}}
  (\bibinfo{publisher}{Elsevier}, \bibinfo{year}{1977}).

\bibitem[{\citenamefont{Amy et~al.}(2013)\citenamefont{Amy, Maslov, Mosca, and
  Roetteler}}]{amy2013meet}
\bibinfo{author}{\bibfnamefont{M.}~\bibnamefont{Amy}},
  \bibinfo{author}{\bibfnamefont{D.}~\bibnamefont{Maslov}},
  \bibinfo{author}{\bibfnamefont{M.}~\bibnamefont{Mosca}}, \bibnamefont{and}
  \bibinfo{author}{\bibfnamefont{M.}~\bibnamefont{Roetteler}},
  \bibinfo{journal}{Computer-Aided Design of Integrated Circuits and Systems,
  IEEE Transactions on} \textbf{\bibinfo{volume}{32}}, \bibinfo{pages}{818}
  (\bibinfo{year}{2013}).

\bibitem[{\citenamefont{Gosset et~al.}(2014)\citenamefont{Gosset, Kliuchnikov,
  Mosca, and Russo}}]{gosset2014algorithm}
\bibinfo{author}{\bibfnamefont{D.}~\bibnamefont{Gosset}},
  \bibinfo{author}{\bibfnamefont{V.}~\bibnamefont{Kliuchnikov}},
  \bibinfo{author}{\bibfnamefont{M.}~\bibnamefont{Mosca}}, \bibnamefont{and}
  \bibinfo{author}{\bibfnamefont{V.}~\bibnamefont{Russo}},
  \bibinfo{journal}{Quant. Inf. and Comp.} \textbf{\bibinfo{volume}{14}},
  \bibinfo{pages}{1261} (\bibinfo{year}{2014}).

\end{thebibliography}

\end{document}